\documentclass[12pt]{article}

\usepackage{setspace}
\usepackage{bm}
\usepackage{indentfirst}
\usepackage{hyperref}   
\usepackage{multicol}
\usepackage{url}       
\usepackage{amsthm,amsmath, amsfonts}

\usepackage{subcaption}
\usepackage{multirow}
\usepackage{booktabs}  
\usepackage{natbib}
\setcitestyle{authoryear,open={},close={}}

\usepackage{algorithm}
\usepackage{algpseudocode}
\usepackage{pifont}
\usepackage{dsfont}
\usepackage{graphicx}
\usepackage[toc,page]{appendix}
\usepackage[utf8]{inputenc}
\usepackage{algorithm}
\newtheorem{theorem}{Theorem}[section]

\newtheorem{remark}{Remark}[section]
\newtheorem{proposition}{Proposition}[section]

\newtheorem{asu}{Assumption}[section]

\newcommand{\epc}{\hspace{1pc}}

\newcommand{\bg}{\bm{g}}

\newcommand{\Eb}{\mathbf{E}}

\newcommand{\Vb}{\mathbf{V}}

\newcommand{\Xb}{\mathbf{X}}

\newcommand{\argmin}{\mathop{\mathrm{argmin}}}

\def\lbegar{$$\left\{ \begin{array}{lll}}
\def\rendar{\end{array} \right.$$}
\def\rendarp{\end{array} \right..$$}
\def\rendarv{\end{array} \right.,$$}

\def\begar{$$\begin{array}{lll}}
\def\endar{\end{array}$$}
\def\begarlab{\begin{equation} \begin{array}{lll} \label}
\def\endarlab{\end{array} \end{equation}}
\def\argmax{\text{argmax}}
\def\argmin{\text{argmin}}

\def\ds1{{\mathrm{1 \hspace{-2.6pt} I}}}

\def\calA{{\cal A}}

\def\calD{{\cal D}}

\def\calF{{\cal F}}
\def\calH{{\cal H}}

\def\calM{{\cal M}}
\def\calN{{\cal N}}

\def\calT{{\cal T}}

\def\calV{{\cal V}}

\def\calX{{\cal X}}

\def\bg1{{g_1}}%\bar

\def\Z1N{Z_1^N}

\def\X1N{X_1^{N}}

\def\l0{{\lambda_0'}}

\def\sign{\text{sign}}

\usepackage{authblk}

\begin{document}
	\title{On  Robustness of Individualized Decision Rules}
%	\author{
%	Zhengling Qi, Jong-Shi Pang and Yufeng Liu\thanks{Zhengling Qi  is Assistant Professor, Department of Decision Sciences, George Washington University. E-mail: qizhengling@email.gwu.edu. Jong-Shi Pang is Epstein Family Chair and Professor, Department of Industrial and Systems Engineering, University of Southern California, LA. E-mail: jongship@usc.edu. Yufeng Liu is Professor, Department of Statistics and Operations Research, Department of Genetics, Department of Biostatistics, Carolina Center for Genome Science, Lineberger Comprehensive Cancer Center, University of North Carolina at Chapel Hill, NC 27599, USA. E-mail: \href{mailto:yfliu@email.unc.edu}{yfliu@email.unc.edu}.}}
\author{
	Zhengling Qi, Jong-Shi Pang and Yufeng Liu\thanks{Zhengling Qi  is Assistant Professor, Department of Decision Sciences, George Washington University. E-mail: qizhengling@email.gwu.edu. Jong-Shi Pang is Epstein Family Chair and Professor, Department of Industrial and Systems Engineering, University of Southern California, LA. E-mail: jongship@usc.edu. Yufeng Liu is Professor, Department of Statistics and Operations Research, Department of Genetics, Department of Biostatistics, Carolina Center for Genome Sciences, Lineberger Comprehensive Cancer Center, University of North Carolina at Chapel Hill, NC 27599, USA. E-mail: \href{mailto:yfliu@email.unc.edu}{yfliu@email.unc.edu}.
The authors would like to thank the Editor, the Associate Editor, and reviewers, whose helpful comments and suggestions led to a much improved presentation.
Yufeng Liu's research was supported in part by NSF grants DMS-1821231 and DMS-2100729, and NIH grant R01GM126550.
	}}
\date{}

%\affil[1]{University of North Carolina, Chapel Hill}
%\affil[*]{E-mail: yfliu@email.unc.edu}

\maketitle
\begin{abstract}
	With the emergence of precision medicine, estimating optimal individualized decision rules (IDRs) has attracted tremendous attention in many scientific areas. Most existing literature has focused on finding optimal IDRs that can maximize the expected outcome for each individual. Motivated by complex individualized decision making procedures and the popular conditional value at risk (CVaR) measure, we propose a new robust criterion to estimate optimal IDRs in order to control the average lower tail of the individuals' outcomes. In addition to improving the individualized expected outcome, our proposed criterion takes risks into consideration, and thus the resulting IDRs can prevent adverse events. The optimal IDR under our criterion can be interpreted as the decision rule that maximizes the ``worst-case" scenario of the individualized outcome when the underlying distribution is perturbed within a constrained set. An efficient non-convex optimization algorithm is proposed with convergence guarantees. We investigate theoretical properties for our estimated optimal IDRs under the proposed criterion such as consistency and finite sample error bounds. Simulation studies and a real data application are used to further demonstrate the robust performance of our methods. Several extensions of the proposed method are also discussed.
	
	\noindent
	{\bf Keywords and Phrases}: Conditional value at risk; Individualized decision rules; Non-convex optimization; Robustness; Tail controls
\end{abstract}

\section{Introduction}
Decision making is a long standing research problem in many scientific areas, ranging from engineering, management science, to statistics. In the era of big data, the traditional ``one fits all" decision rules are no longer ideal in many applications due to data heterogeneity. A decision rule that works for certain individuals may not necessarily work for others. Motivated by this, it is desirable to make individualized decision rules (IDRs) that map from individual characteristics into available decision assignments.  Developing effective IDRs has a wide range of applications. For example, a credit card company hopes to send a special offer for each targeted customer tailoring to their personal needs. An epidemiologist needs to decide whether to deliver a vaccine plan to a specific region in order to prevent the spread of diseases. In medical applications, IDRs can be developed for better prevention and treatment methods that are tailored to each individual patient. Developing optimal IDRs is one of the main goals of precision medicine, also known as personalized medicine.

 Most existing methods in the literature, from the data analytic perspective, are focused on estimating the optimal IDR that can maximize the expected outcome or minimize the expected loss for each individual (\cite{qian2011performance, manski2004statistical}). For example, we may want to learn an IDR $d$ that maps an individual's covariate $\Xb \in \calX$ into a binary decision space $\calA = \{1, -1\}$, i.e., $d: \calX \rightarrow \calA$, in order to maximize the expectation of $R(d)$ or a utility function $u$ of $R(d)$. Here $R(d)$ is the random outcome under IDR $d$, which will be formally defined in Section 2. Then the problem can be mathematically formulated as the following optimization problem:
\begin{equation}\label{basic}
	\max_{d \in \calD_0} \epc \Eb[R(d)],
\end{equation} 
where $\calD_0$ is a class of all IDRs. If some utility function $u$ is used, then one could replace the objective function in \eqref{basic} with $\Eb[u(R(d))]$. The equivalent form of Problem \eqref{basic} is 
\begin{equation}\label{epi}
	\begin{aligned}
	\max_{t \in \mathbb{R}, d \in \calD_0} & \epc t \\
	\text{subject to}& \epc t \leq \Eb[R(d)],
	\end{aligned}
\end{equation}
by using the hypographical representation (\cite{rockafellar2009variational}). 

According to \eqref{epi}, it may be reasonable to restrict the expected outcome larger than some threshold when stochastic ups and downs of outcome $R(d)$ can be safely averaged out. Then by implementing the IDR obtained by \eqref{epi}, we can guarantee on average the outcome will be at least as good as some threshold $t$. However, in practice, practitioners usually want to have a safe margin to protect against undesired outcomes, especially when the lower tails of the outcome are more important. For instance, in reliability engineering (\cite{rockafellar2010buffered}), people are often interested in controlling the failure probability of some designed systems or structures such as buildings and bridges, instead of expected failure time. Here the selection of designs can be interpreted as IDRs, which are based on environmental conditions, existing materials, etc. 
%In insurance management, in addition to maximizing the average profits, a credit card company may want to avoid bad debt expense caused by sending an inappropriate offer to some customers. In epidemiology, an epidemiologist may want to prevent diseases spreading by choosing the vaccine plan that can control most population's risk. 
In precision medicine, sometimes the gain in the expected outcome may be very little by comparing two treatments while the tail of the potential outcome distribution is of direct interest (\cite{wang2017quantile}). Suppose two drugs are used to improve CD4 T cell amount of AIDS patients. Since the normal range of CD4 cells is $500$-$1500$, then in practice the event of a certain subject $\{R_i(d) = 600 \}$ may be treated as good as the event of  $\{R_i(d) =800 \}$. In contrast, the event $\{R_i(d) = 50 \}$ may be considered much worse than $\{R_i(d) = 400 \}$. Hence only using \eqref{epi} to search for the best IDR may not be sufficient when the tails are more important or the variability of the outcomes needs to be controlled. In this case, one may consider using the truncated mean with a fixed cutoff as a criterion for better decision making, i.e., $\Eb[\min(R(d), L)]$, where $L$ is a fixed cutoff. However, in general, the cutoff $L$ is unknown and can be difficult to determine. Even though $L$ can be pre-determined in some ideal case, when the scale of $R(d)$ is changed or there is some measurement error or batch effect in the observed $R_i(d)$, the final $L$ needs to be properly modified, which could still be hard to choose. In addition, for many applications especially medical studies, normal ranges are usually determined by quantiles of healthy people. Therefore it may be more reasonable to consider the expected outcome lower than a certain quantile when evaluating IDRs.
 In this paper, motivated by the conditional value at risk (CVaR) used extensively in finance and risk management (\cite{rockafellar2000optimization}), we propose a new criterion that considers average lower tails of outcome to evaluate IDRs. The resulting IDR under our proposed criterion can optimize the outcome of each individual and provide a safe margin against adverse events jointly. Our work is closely related to the recent paper by \cite{qicuiliupang19}, where they established a mathematical framework to study IDRs under general risk measures. In this work, we are more specific on the lower tail of outcomes for making decisions and develop a thorough statistical analysis. We propose an efficient optimization algorithm for our specific problem with better convergence guarantee than that of \cite{qicuiliupang19}. In addition, several practical extensions are also discussed.

\subsection{Related Literature}
There is an increasing body of literature in learning optimal IDR under the framework of \eqref{basic}. The literature spans across various fields such as statistics, economics, and machine learning. Along the line of statistics literature, most existing methods can be roughly divided into two categories: model based methods and direct search methods. Q-learning (\cite{watkins1989learning}, \cite{murphy2005generalization}, \cite{schulte2014q}) and A-learning (\cite{murphy2003optimal}, \cite{robins2004optimal}, \cite{shi2018high}) are two representative model based methods. Other variants include \cite{fan2016concordance}, \cite{gunter2011variable}, etc. For direct search methods, by viewing IDR problems as a weighted classification problem, \cite{zhao2012estimating} proposed to use the weighted support vector machine method to estimate the optimal IDR. Following that, various types of machine learning methods were proposed, such as \cite{liu2016robust}, \cite{zhou2017residual}, \cite{zhao2015doubly}, \cite{cui2017tree}, \cite{tao2016adaptive}, \cite{laber2015tree}, \cite{zhang2015using}, \cite{chen2018estimating}, \cite{zhangss20}. In addition,  \cite{tian2014simple}, \cite{qi2017} and \cite{qi2018} proposed to use regression methods to directly estimate the optimal IDR. Recently, \cite{wang2017quantile} proposed to use the quantile of $R(d)$ as a criterion to search for best IDRs, which is closely related to this paper. Their method can help to obtain robust optimal IDRs by controlling the lower quantile. However, it can be unstable when the potential outcome distribution is discrete. More importantly, as we mentioned in the example of CD4 T cells, lower tails of the outcome should be treated differently: if both $\{R_i(d) = 50\}$ and $\{R_i(d) = 400\}$ are below the $25\%$ quantile, the first event should be considered much worse than the latter one. We will discuss this further in Section 2. Other quantile-based methods include \cite{linn2017interactive}, \cite{xiao2019robust} and \cite{fang2021fairness}.

In the econometrics literature, \cite{manski2004statistical} provided a comprehensive regret analysis on the estimation of optimal IDRs with some connection to statistical decision theory (\cite{savage1951theory}). Later on, exact finite sample regret analysis was established by \cite{stoye2009minimax} and \cite{tetenov2012statistical}. Under smooth parametric and semi-parametric settings, \cite{hirano2009asymptotics} investigated asymptotic optimality and large sample properties of optimal IDRs. Other related work includes \cite{chamberlain2011bayesian,bhattacharya2012inferring,kasy2016partial}. Recently, \cite{kitagawa2018should} and  \cite{athey2017efficient} established rate-optimal regret bounds for learning optimal IDRs. All of these existing developments are based on expected outcome or expected utility. \cite{dehejia2008ate} studied the risk aversion of treatment effect evaluation by using the mean-variance trade off criterion.

We would like to point out the increasing literature of learning optimal IDRs in the machine learning community, which is often referred as batch learning from bandit feedback, such as \cite{Beygelzimer:2009:OTL:1557019.1557040,dudik2011doubly,JMLR:v16:swaminathan15a,dudik2011doubly,NIPS2018_8105}. Finally,  in the reinforcement learning literature, CVaR has been used as constraints in Markov decision processes such as \cite{Tamar:2015:OCV:2888116.2888133} and \cite{ Chow:2017:RRL:3122009.3242024}.

\subsection{Main Contributions and Outline}
The main contributions of this paper can be summarized as follows. We leverage the CVaR criterion and  propose a robust criterion to directly estimate the optimal IDR that can  improve the expected outcome while simultaneously controlling the average lower tails of the outcome. We discuss several important properties of our criterion and its practical usage by providing safety protection for implementing optimal IDRs. An efficient non-convex optimization algorithm is proposed to compute the solutions with a convergence guarantee. We establish several important theoretical properties of our estimator under the proposed criterion similar to the regret analysis in \cite{zhou2017residual, athey2017efficient}, where the difference is to handle additional estimation error due to the special structure of the CVaR criterion.

The remainder of this paper is organized as follows. In Section 2, supplementing the previous value function framework, we introduce a new criterion to estimate the optimal IDR by using the concept of CVaR in risk management.  We present several properties of our proposed criterion. In Section 3, we discuss our statistical estimation procedure to compute optimal IDRs under our proposed criterion. An efficient non-convex optimization algorithm is presented by using some recent developments in difference of convex algorithms (DCA).
In Section 4, we establish several important theoretical properties of our method based on statistical learning theory. We demonstrate our method via extensive simulation studies and a real data application in Sections 5. In Section 6, we discuss several extensions of our proposed criterion from the perspectives of algorithm and modeling. Some technical results are provided in the supplementary materials.
\vspace{-0.2in}
\section{Robust Criteria to Estimate Optimal IDRs}
\subsection{Notation and Basic Settings}
We discuss our IDR problem under the potential outcome framework (\cite{rubin1974estimating}).
We use $R$ to denote the outcome after receiving treatment $A$. We consider a binary-treatment setting and encode $A$ as either $1$ or $-1$. The treatment space is denoted by $\calA$, i.e., $\calA = \left\{-1, 1\right\}$. %Without loss of generality, we assume a larger outcome $R$ is more preferred. 
Let $\Xb \in \calX$ denote the $p$-dimensional random vector for covariates. Here $\calX \subset \mathbb{R}^p$ is the covariate space. Throughout this paper, we make the following three standard assumptions.

%Suppose there is a random sample of size $n$ from population. For each subject in this sample, we observe covariate information $\Xb_i \in \calX$, where $\calX \subseteq \mathbb{R}^p$, for $i = 1, \cdots, n$. Each subject is associated with a pair of potential outcomes, $(R_i(1), R_i(-1))$, which are scalars. We can only observe either $R_i(1)$ or $R_i(-1)$ based on which treatment being assigned. Here $A_i \in \calA =  \{1, -1\}$ denotes the treatment that the $i$-th subject has received. Therefore, for each subject, we observe $(\Xb_i, A_i, R_i)$, where $R_i = R_i(A_i)$. By doing that, we assume the stable unit treatment value assumption (SUTVA) (\cite{rubin1990comment}). We define $\pi(A | \Xb)$ as the probability of a subject being assigned treatment $A$ given the covariates and denote $P^d$ as the probability measure when the treatment follows the decision rules. In addition,
\begin{asu}\label{consistency}
 	 $R = R(A)$.
\end{asu}
\begin{asu}\label{confounder}
	 $R(a) \; \bot \; A \, | \,  \Xb$ for any $a \in \calA$, where $\bot$ represents independence.
\end{asu}
\begin{asu}\label{overlap}
	$ \pi(a | \Xb) \geq c$ almost surely for any $a \in \calA$ and some positive constant $c$. 
\end{asu}
For simplicity, we assume the random outcome $R$ has a bounded support. This assumption can be relaxed by the high-order moment condition. See \cite{athey2017efficient} for more details. Without loss of generality, we assume that the larger $R$ indicates the better condition an individual is in. Throughout this paper, we consider a randomized experiment and therefore the propensity score $\pi(A | \Xb)$ is known. For the observational studies, the proposed method can be applied as well by estimating the propensity score using various methods such as the logistic regression. An IDR $d$ is defined as a measurable function mapping from the covariate space $\calX$ into the treatment space $\calA$. 
%For any IDR $d$, define $P^d$ to be the probability measure when the treatment $A$ follows the decision rules $d$.  Furthermore, let $P$ be the probability distribution of a random triplet $(\Xb, A, R)$.
% under which the likelihood of $(\Xb, A, R)$ is defined as $f_0(x)\pi(a | x)f_1(r | x, a)$. In particular, $f_0(x)$ is the probability density function of $\Xb$ and $f_1(r | x, a)$ is the conditional probability density function of $R$ given $(A, \Xb)$.
%
%Then the probability density function under $P^d$ is defined as $f_0(x)\mathbb{I}(a = d(x))f_1(r|x, a)$, where $\mathbb{I}(\bullet)$ is the indicator function.
 We also let $L^r(\calT, \calF_1, \mathbb{P})$ be the space of all measurable functions such that $\int_{T \in \calT} |f(T)|^r \text{d}P^d < \infty $, where $\calF_1$ is the  $\sigma$-field generated by $\calT := \calX \times \calA \times \mathbb{R}$ and $\mathbb{P}$ is the corresponding probability measure.
\subsection{Expected Value Function Framework}
Before introducing our new criterion and methods, we first present the value function framework used by the most existing methods, such as (\cite{manski2004statistical}) and (\cite{qian2011performance}). The value function of an IDR $d$ is defined as
\begin{equation}\label{value fun}
\begin{aligned}
V(d) & = \Eb[R(d)] = \Eb_\Xb\left[\Eb[R(d) | \Xb]\right] =\Eb_\Xb\left[\Eb[R(d) | \Xb, A = d(\Xb)]\right]\\ 
& = \Eb_\Xb\left[\Eb[R| \Xb, A = d(\Xb)]\right]
=\Eb\left[\frac{R\mathbb{I}(A = d(\Xb))}{\pi(A | \Xb)}\right],
\end{aligned}
\end{equation}
where the first line is based on Assumption \ref{confounder}, the first equality in the second  line is based on Assumption \ref{consistency} and the last equality relies on Assumption \ref{overlap}. Based on this value function, an optimal IDR $d_{v}$ under the mean criterion is defined as 
\vspace{-0.2in}
\begin{equation*}\label{optimal d}
	d_v \in \argmax_{d \in \calD_0} V(d),
\end{equation*}
which is equivalent to
\begin{equation*}\label{optimal d 2}
	d_v(\Xb) \in \argmax_{a \in \calA} \,  \Eb\left[R \, | \, \Xb, A = a\right],
\end{equation*}
almost surely. It is observed that under the value function framework, the optimal IDR is to select the treatment with the largest expected outcome among all treatments for each individual.

Despite much progress towards developing optimal IDRs in the intersection of statistics, econometrics, and machine learning, only focusing on obtaining the largest expected outcome for each individual can be too restrictive and sometimes may not be even safe. For example, doctors may want to know whether a treatment does the best to improve the worst scenario, in particular for a patient with high risk. Without such a risk consideration, it may lead to serious consequences, such as exacerbation or hospitalization in practice. Similar concerns may happen in the credit card company, where the ``best" policy should not only improve the expected profit, but also reduce the chance of incurring a heavy loss. This motivates us to control risk exposure associated with the corresponding decision rules, in addition to maximizing the expected outcome for each individual.
\vspace{-0.2in} 
\subsection{Conditional Value at Risk}
It is natural to consider some robust metrics such as quantiles of $R$ given $\Xb$ and $A$ to measure the effect of a treatment (\cite{wang2017quantile}). The corresponding optimal IDR $\widetilde{d}$ under the quantile may be defined as
\begin{equation}\label{quantile optim}
	\widetilde{d} \in \argmax_{d \in \calD_0} \, Q_\gamma(R(d)),
\end{equation}
where $Q_\gamma(R(d)) = \inf \{y: P(R(d) \leq y) \geq 1 - \gamma  \}$ and $\gamma \in (0, 1)$. Analogous to Problems \eqref{basic} and \eqref{epi}, an equivalent form of Problem \eqref{quantile optim} can be constructed as follows:
\begin{equation}\label{epi quantile}
	\begin{aligned}
	\max_{t \in \mathbb{R}, d \in \calD_0} & \epc t\\
	\text{subject to} & \epc t \leq Q_\gamma(R(d)).
	\end{aligned}
\end{equation} 
Based on the above representation, roughly speaking, the constraint set $\{t \leq Q_\gamma(R(d)) \}$ implies that $(1-\gamma) \times 100\%$ of the  population under the given IDR are controlled to be at least as good as a certain threshold $t$. Thus under the quantile criterion, one can obtain a robust IDR that can improve almost $(1-\gamma) \times 100\%$ of the population to some extent.

There are several potential drawbacks of using quantile in IDR problems. First of all, using the quantile criterion treats all the outcomes lower than $Q_\gamma(R(d))$ as the same. However, as we discussed in the aforementioned example, the CD4 T cell $\{R_i(d) = 400\}$ below the normal level is considered to be much better than $\{R_i(d) = 50\}$, therefore they should be treated differently in practice. Secondly, $Q_\gamma(R(d))$ is generally not continuous in $\gamma$, which may cause instability. For instance, if the outcome distribution is discrete and there is a small change in $\gamma$, the resulting optimal $\widetilde{d}$ may change significantly. Lastly, from the computational perspective, the quantile makes the optimization Problem \eqref{quantile optim} hard to solve and thus may limit its use in practice.

In order to address the drawbacks of using quantiles, we propose to use the conditional value at risk (CVaR), also known as expected tail loss, which was proposed by \cite{artzner1999coherent} in risk management. Consider a continuous random variable $Y$. Then the $\gamma$-CVaR of $Y$ may be defined as
\begin{equation}\label{def CVaR}
S(F_Y) := \frac{1}{\gamma}\Eb[Y\mathbb{I}(Y \leq Q_\gamma(F_Y))],
\end{equation}
where $F_Y$ is the corresponding probability distribution of $Y$. Based on this definition, the $\gamma$-CVaR can be interpreted as a truncated mean lower than $\gamma$-quantile of $Y$. For the general setting, instead of assuming a continuous distribution of $Y$, $\gamma$-CVaR is defined as an optimal value of a concave maximization problem by the celebrated work of \cite{ben1986expected,rockafellar2000optimization}, which is defined as follows:
\begin{equation}\label{CVaR compute}
S(F_Y) := \sup_{\alpha \in \mathbb{R}} \left\{ \alpha - \frac{1}{\gamma}\Eb[(\alpha - Y)_+]\right\},
\end{equation}
where $[t]_+ = \max(0, t)$. The leftmost of the optimal solution set to \eqref{CVaR compute} is $Q_\gamma(F_Y)$ (\cite[~Theorem 1]{rockafellar2000optimization}). Then one can see that the definition in \eqref{CVaR compute} is equivalent to \eqref{def CVaR} when the outcome distribution of $Y$ is continuous. We also remark that  $Y$ is often referred to a loss in the finance literature. Here we call $Y$ an outcome/reward in order to be consistent with our problem setting.

Lastly, we would like to point out that $\gamma$-CVaR has several nice properties discussed by \cite{artzner1999coherent} and it is in general preferable to the quantile measure (\cite{sarykalin2008value}). In particular,
based on the interpretation of \eqref{def CVaR}, $\gamma$-CVaR considers average outcomes lower than $\gamma$-quantile, which treats lower tails of outcome differently. This exactly satisfies our purpose. The concave maximization formulation in \eqref{CVaR compute} demonstrates the continuity of $\gamma$-CVaR with respect to $\gamma$, which provides more stability measure compared with the quantile one. Furthermore, $\gamma$-CVaR is considered to be more computationally efficient than the quantile criterion because of the concave maximization formulation. Finally, \cite{pflug2000some} and \cite{rockafellar2002conditional} showed that $S(F_Y) \leq Q_\gamma(F_Y)$, suggesting that $\gamma$-CVaR is more conservative than $\gamma$-quantile. This also implies that a larger $\gamma$-CVaR of a random outcome indicates a larger $\gamma$-quantile. Clearly the reverse inequality does not necessarily hold. All these nice properties motivate us to use CVaR in the IDR problems. 
%For simplicity, we assume $R$ is absolutely continuous with respect to $P^d$, but the result remains valid for the discrete distribution with little change.
 For the related theoretical discussion about CVaR, we refer to \cite{rockafellar2002conditional} and the references therein.
\subsection{A New Robust Criterion for IDR Problems}
We borrow the concept of CVaR in order to conduct risk control to obtain a robust optimal IDR that can improve the lower tails of outcomes. Specifically, we define the following decision-rule based $\gamma$-CVaR criterion as
 \begin{equation}\label{M0}
 	M_{0, \gamma}(d) := \sup_{ \alpha \in \mathbb{R}} \left\{\alpha - \frac{1}{\gamma}\Eb[(\alpha - R(d))_+] \right\}.
 \end{equation} 
Note that the difference between the criterion \eqref{M0} and the original CVaR is to let the outcome $R$ depend on IDR $d$. The proposed criterion is not restricted to the continuous outcomes but can also be used in the discrete outcome cases, thus providing a broad application. In addition,  if a small value of the outcome is more desirable (i.e., measuring the loss instead of the reward), one can replace $R(d)$ by $-R(d)$. All properties preserve. Given the coherent property of CVaR shown in \cite{artzner1999coherent}, we have the following proposition for $M_{0,\gamma}(d)$.
 \begin{proposition}\label{propertyofM0}
 	The following properties of $M_{0,\gamma}(d)$ hold.
 	\begin{itemize}
 		\item[(a)] If $R$ is shifted by a  constant $c$, then $M_{0,\gamma}(d)$ is also shifted by the same constant $c$;
 		\item[(b)] If $R$ is multiplied by a positive constant $c$, then the corresponding $M_{0,\gamma}(d)$ is also multiplied by the same constant $c$;
 		\item[(c)] Given two IDRs $d_1$ and $d_2$, if $R(d_1) \leq R(d_2)$ almost surely, then
 		\begin{equation*}\label{monotonicity}
 		M_{0,\gamma}(d_1) \leq M_{0,\gamma}(d_2);
 		\end{equation*}
% 		\item[(d)] Suppose we use $d_1$ with probability $w$ and and  use $d_2$ with probability $(1-w)$, where $w \in [0, 1]$. Then we have
% 		\begin{equation}
% 			M_{0,\gamma}(\bar{d}) \geq wM_{0,\gamma}(d_1) + (1-w)M_{0,\gamma}(d_2),
% 		\end{equation}
% 		where $M_{0,\gamma}(\bar{d}) = \sup_{ \alpha \in \mathbb{R}} \left\{\alpha - \frac{1}{\gamma}\Eb[(\alpha - (wR(d_1) + (1-w)R(d_2)))_+] \right\}$;
 		\item[(d)] Given IDR $d$, $M_{0,\gamma}(d) \leq \min\{V(d), Q_\gamma(R(d))\}$.
 	\end{itemize}
 In addition, if outcome $R = c$ almost surely, then $M_{0,\gamma}(d) = c$.
 \end{proposition}
\begin{remark}
 Proposition \eqref{propertyofM0} justifies the use of $M_{0,\gamma}(d)$. In particular, $(a)$ and $(b)$ demonstrate that $M_{0,\gamma}(d)$ is not affected by a constant shift or multiplication. $(c)$ implies that if one IDR is no worse than the other, $M_{0,\gamma}(d)$ preserves the preferences.
%	The property $(d)$ highlights the importance of concavity in some sense. Basically, suppose we have two IDRs. One is to use $d_1$ on 	$w \times 100\%$ of the population and $d_2$ on the remaining population. The other is to use $d_1$ with probability $w$ and $d_2$ with probability $(1-w)$.
%	Our criterion $M_{0,\gamma}$ always prefers the first IDR. (MARK)For example, if treatment $1$ is always better than treatment $-1$ for each individual, then it is always safer to let $w \times 100\%$ of the population use treatment $1$ and the remaining use the other treatment instead of letting whole population use one particular treatment with some probability. 
	The last property  indicates that $M_{0,\gamma}(d)$ is more conservative than the expected outcome and the quantile criterion when evaluating an IDR $d$.
\end{remark}
If the distribution of $R(d)$ is continuous, then we could rewrite \eqref{M0} as 
\begin{equation}\label{mix fun 1}
M_{0,\gamma}(d) = \frac{\Eb^d[R\mathbb{I}(R \leq Q_\gamma(R(d)))]}{\gamma}.
\end{equation}
Note that
\begin{equation}\label{average risk}
\begin{aligned}
P(R(d) < \alpha)
= \Eb[\sum_{a \in \calA} \mathbb{I}(d(\Xb) =a)P(R < \alpha | \Xb, A = a)]= \Eb[P(R < \alpha | \Xb, A = d(\Xb))].
\end{aligned}
\end{equation}
Then $Q_\gamma(R(d))$ can be further expressed as 
\begin{equation}\label{Q_gamma}
Q_\gamma(R(d)) = \inf \left\{ \alpha \, | \, \Eb[P(R < \alpha | \Xb, A = d(\Xb))] \geq \gamma \right\},
\end{equation}
which can be interpreted as the average $\gamma$-quantile of $R$ under the decision rule $d$. Correspondingly $M_{0,\gamma}(d)$ can be understood as the $\gamma$-average CVaR. 

According to Proposition \ref{propertyofM0} (d), $M_{0,\gamma}(d)$ can be regarded as a lower bound of $V(d)$ and $Q_\gamma(R(d))$. Based on this, maximizing $M_{0,\gamma}(d)$ can thus potentially improve both $V(d)$ and $Q_\gamma(R(d))$. Then the optimal IDR under our proposed robust criterion $M_{0,\gamma}(d)$ is defined as
\begin{equation}\label{optimal d under M0}
d_{0, \gamma} \in \argmax_{d \in \calD_0} \ M_{0,\gamma}(d).
\end{equation}
The optimal IDR with respect to $M_{0,\gamma}(d)$ is to select a treatment/decision with the largest $\gamma$-average CVaR. Moreover, if we again use hypographical representation similar as those in \eqref{epi} and \eqref{epi quantile}, then we could formulate \eqref{optimal d under M0} as the following constraint optimization problem.
\begin{equation}\label{constrain M0}
\begin{aligned}
\max_{d \in \calD_0, t \in \mathbb{R}} & \epc  t\\
\text{subject to} & \epc \sup_{ \alpha \in \mathbb{R}} \left\{\alpha - \frac{1}{\gamma}\Eb[(\alpha - R(d))_+] \right\} \geq t.
\end{aligned}
\end{equation}
Based on the formulation \eqref{constrain M0}, $d_{0, \gamma}$ can be interpreted as the best IDR with the average lower tail of the outcome being at least as good as some certain threshold. Besides, based on statement (d) in Proposition \eqref{propertyofM0}, the constraint set in \eqref{constrain M0} implies $\min\{V(d), Q_\gamma(R(d))\} \geq t$. Therefore, by using $M_{0,\gamma}(d)$, the resulting optimal $d_{0, \gamma}$ can guarantee both $(1-\gamma) \times 100\%$ of the population's outcomes and the expected outcome no worse than some threshold $t$.

By a similar derivation as that in \eqref{value fun}, we have the following proposition.
 \begin{proposition}\label{M0equiv}
 	$M_{0,\gamma}(d) =\sup_{ \alpha \in \mathbb{R}} \left\{\alpha - \frac{1}{\gamma}\Eb[\frac{\mathbb{I}(A = d(\Xb))}{\pi(A | \Xb)}(\alpha - R)_+] \right\}.$
 \end{proposition}
%\subsubsection{Average Lower Tail}

The definition in \eqref{M0} and Proposition \eqref{M0equiv} give us an efficient way to compute the optimal IDR $d_{0, \gamma}$ and $\alpha_0$ via jointly  optimizing
\begin{equation}\label{M0compute}
	(d_{0, \gamma}, \alpha_0) \in \underset{\alpha \in \mathbb{R}, d \in \calD_0}{\argmax} \left\{ \alpha - \frac{1}{\gamma}\Eb^d[(\alpha - R)_+] \right\}.
\end{equation}
\subsection{Dual Representation}
Note that $M_{0,\gamma}(d)$ involves concave maximization with respect to $\alpha$. Thus it would be useful to investigate its dual representation by making use of convex duality theory in optimization (e.g., \cite{rockafellar1974conjugate}). To begin with, we first define the following set:
\begin{equation*}
	\calV_0^d := \{ V \in L^1(\calT, \calF_1, P^d) \, | \, \Eb^d[V] = 1, \, 0 \leq V(\omega) \leq \frac{1}{\gamma}, \; \text{for} \; \omega \in \calT \, \text{almost surely} \},
\end{equation*}
where $P^d$ is the probability measure under the IDR $d$.
%and
%\begin{equation}
%\calW^d_2 := \{W\in L^1(\calT, \calF_1, P^d) \, | \, \epsilon_1 \leq W(\omega_1) \leq \varepsilon_2 \; \text{for almost sure} \; \omega_1 \in \calT, \, \Eb[W| \Xb, A = d(\Xb)] = 1   \},
%\end{equation}
We have the following theorem that gives the dual representation of $M_{0,\gamma}(d)$. 
%\begin{equation}
%	\begin{aligned}
%	M_2(d) & = \sup_{\alpha \in L^1(\calX, \Xi, P_\Xb)} \left\{(1-\tau)\Eb^d[R] + \tau  \Eb^d[\alpha(\Xb) - \frac{1}{\gamma}(\alpha(\Xb) - R)_+] \right\} \\
%	& = \sup_{\alpha \in L^1(\calX, \Xi, P_\Xb)} \left\{ \Eb^d[\alpha(\Xb)] + (1-\tau)\Eb^d[R- \alpha(X)] - \tau  \Eb^d[ \frac{1}{\gamma}(\alpha(\Xb) - R)_+] \right\} \\
%	& = \sup_{\alpha \in L^1(\calX, \Xi, P_\Xb)} \left\{ \Eb^d[\alpha(\Xb)] + (1-\tau)\Eb^d[(R- \alpha(X))_+] - \Eb^d[ (1-\tau +\frac{\tau}{\gamma})(\alpha(\Xb) - R)_+] \right\} \\
%	\end{aligned}
%\end{equation}
\begin{theorem}\label{thm: duality}
	$M_{0,\gamma}(d) = \inf_{V \in \calV_0^d} \Eb^d[VR]$.
\end{theorem}
\textbf{Dual representation of $M_{0,\gamma}(d)$}: According to the dual representation of $M_{0,\gamma}(d)$ and its proof in the appendix, we can define a conditional probability measure $P^\Vb(B) = \int_B V dP^d$ for any measurable set $B \in \calF_1$, where $V \in \calV^d_0$. Then $V = \frac{dP^\Vb}{dP^d}$. Define
	\[
	\zeta_\gamma(u) =
	\begin{cases}
	0 & \text{if \ $0 \leq u \leq \frac{1}{\gamma}$} \\
	+\infty & \text{otherwise}\\
	\end{cases}.
	\]
Then we can further rewrite $\Eb^d[VR] $ in $M_{0,\gamma}(d)$ as
\begin{equation*}
\begin{aligned}
\Eb^d[RV] & = \Eb^d[R\frac{\text{d}P^\Vb}{\text{d}P^d}] + \Eb\left[\zeta_\gamma(\frac{\text{d}P^\Vb}{\text{d}P^d})\right] = \Eb_\Xb[\Eb_{P^\Vb_{| \Xb}}[R]] + \Eb_\Xb\left[\int \zeta_\gamma(\frac{dP^\Vb}{dP^d}) dP^d\right]  = \Eb_{P^\Vb}[R] + I_{\zeta_\gamma}(\frac{dP^\Vb}{dP^d}), 
\end{aligned}
\end{equation*}
where $I_{\zeta_\gamma}(\cdot)$ can be interpreted as the $f$-divergence distance between $P^\Vb$ and $P^d$. Then $M_{0,\gamma}(d) = \inf_{P^\Vb \ll P^d}\Eb_{P^\Vb}[R] + I_{\zeta_\gamma}(\frac{dP^\Vb}{dP^d})$, where that $u \ll v$ denotes the probability measure $u$ is absolutely continuous with respect to the probability measure $v$. Thus the optimal IDR can also be written as
\begin{equation*}
d_{0, \gamma} \in \argmax_{d \in \calD_0} \left\{\inf_{P^\Vb \ll P^d}\Eb_{P^\Vb}[R] +  I_{\zeta_\gamma}(\frac{dP^\Vb}{dP^d} )\ \right\},
\end{equation*}
which can be interpreted as choosing an optimal decision rule in terms of the worst expected outcome within the $f$-divergence distance from the original distribution $P^d$. 

%	Moreover, if we define $\Omega = \{P^\Qb | P^\Qb_{| \Xb} \ll P^d_{| \Xb}, 0 \leq \frac{dP^\Qb_{| \Xb}}{dP^d_{| \Xb}} \leq \frac{1}{\gamma}  \}$, then $M_{0,\gamma}(d) =\min_{P^\Qb \in \Omega}\Eb_{P^\Qb}[R]$ and $d_2 \in \argmax_{d}\{ \min_{P^\Qb \in \Omega}\Eb_{P^\Qb}[R]\}$, which demonstrates the robustness of our methods from dual representation.
%\begin{remark}
According to our problem setting, define the density under $P^d$ as $f_0(x)\mathbb{I}(d(x) = a) f_1(r | x, a)$. Since $P^\Vb \ll P^d$, then the density under $P^\Vb$ should be $v_0(x)\mathbb{I}(d(x) = a) v(r | x, a)$ for some conditional probability density $v_0(x)$ and $v_1(r | x, a)$. Then we have $\frac{d P^\Vb}{d P^d} =\frac{v_0(x) v_1(r | x, a= d(x))}{f_0(x)f_1(r | x, a= d(x))}$ by the chain rule. Therefore, we can further rewrite $M_{0,\gamma}(d)$ as
\begin{equation}\label{min-max M_1}
M_{0,\gamma}(d) = \inf_{P^\Vb} \left\{\Eb_{P^\Vb}[R] \, | \, P^\Vb \ll P^d, \; 0 \leq  \frac{v_0(\Xb) v_1(r | \Xb, a= d(\Xb))}{f_0(\Xb)f_1(r | \Xb, a= d(\Xb))} \leq \frac{1}{\gamma}, \; \mbox{almost surely} \right\}.
\end{equation}
This gives us a natural link to distributionally robust statistical models that can evaluate a decision rule under ambiguity. Maximizing $M_{0,\gamma}(d)$ over $\calD_0$ is equivalent to identifying an optimal IDR that is robust to the contamination of both outcome $R$ and covariates $\Xb$ characterized by a probability constraint set. Recently, \cite{mo20} investigated distributionally robust ITRs by directly optimizing the expected rewards for the worse ITR within a perturbation set around the training distribution. Their method mainly focuses on potential covariate shifts. In contrast, by using a CVaR-based criterion, our proposed method is robust to distributional changes of both the outcome and covariates.
\section{Statistical Estimation and Optimization}
In this section, we discuss the estimation and optimization procedures for Problem \eqref{M0equiv} given observed data. 
%Before that, we first introduce some definitions related to the algorithm convergence of non-convex optimization problems.
%
%Let $\Phi: \mathbb{R}^n \rightarrow \mathbb{R}$. The directional derivative of $\Phi$ at a point $x \in \mathbb{R}^n$ along the direction $v \in \mathbb{R}^n$ is given by
%\begin{equation}\label{def dd}
%	\Phi'(x, v) = \lim_{\tau \downarrow 0} \frac{\Phi(x + \tau v) - \Phi(x)}{\tau}.
%\end{equation}
%We say $x_0$ is a directional-stationary (d-stationary) point of $\Phi$ on  $\mathbb{R}^n$ if 
%\begin{equation}\label{d stationary}
%	\Phi'(x_0, x - x_0) \geq 0, \ \forall x \in \mathbb{R}^n.
%\end{equation}
%For a directionally differentiable optimization problem, d-stationary points can be viewed as the first order ``sharpest" ones among different kinds of stationary points including Clarke points (\cite{pang2016computing}), and the condition \eqref{d stationary} is the least relaxed among other types of stationarity conditions. In the following subsections, we develop two algorithms to compute d-stationary points of Problems \eqref{mix fun compute} and \eqref{mix fun compute 2} respectively, which is the best we can achieve for non-convex optimization problems in practice.
%\subsection{Estimation of Optimal IDRs  under $M_1(d)$}
The optimization in \eqref{M0compute} can be rewritten as
\begin{equation*}
	\max_{ \alpha \in \mathbb{R}, d \in \calD} \Eb[\frac{ (\alpha - \frac{(\alpha - R)_+}{\gamma})\mathbb{I}(A = d(\Xb))}{\pi(A | \Xb)}],
\end{equation*}
where $\calD$ is some pre-specified classes of decision rules such as the linear ones. Consider the binary treatment setting and let $d(\Xb) = \sign (f(\Xb))$. Suppose we observe independently and identically distributed data $(\Xb_i, A_i, R_i); i = 1, \cdots, n$. Then we can estimate the optimal IDR via the empirical approximation:
\begin{equation}\label{empirical M0}
\max_{ \alpha \in \mathbb{R}, d \in \calD} \frac{1}{n}\sum_{i = 1}^{n} \frac{\mathbb{I}(A_i = \sign(f(\Xb_i)))}{\pi(A_i | \Xb_i)}(\alpha - \frac{(\alpha - R_i)_+}{\gamma}).
\end{equation}
It is well known that optimization over indicator functions is NP hard. Alternatively, similar to \cite{zhou2017residual}, we can replace the $0$-$1$ loss function by the following smooth truncated loss,
\[
S(u) =
\begin{cases}
0 & \text{if $u \leq -\delta$} \\
(1+u/\delta)^2 / 2 & \text{if $0 > u \geq -\delta$} \\
1 - (1-u/\delta)^2/2 & \text{if $\delta > u \geq 0$}\\
1 & \text{if $u > \delta$},
\end{cases}
\]
 and then use a functional margin representation to express $\mathbb{I}(A_i = \sign(f(\Xb_i))$ as $\mathbb{I}(A_i f(\Xb_i) > 0)$ for each $i$.
The corresponding function plot of $S(u)$ is shown in Figure \ref{fig:surrogate} with $\delta = 1$. From the plot, we can see that the smooth approximation $S(u)$ is very close to the  $0$-$1$ loss. %, which does not affect our approximation to the $0$-$1$ loss. 
The parameter $\delta$ can control the closeness of this approximation. In practice, we can simply choose $\delta = 1$.
\begin{figure}[H]
	\begin{center}
		\includegraphics[scale = 0.65]{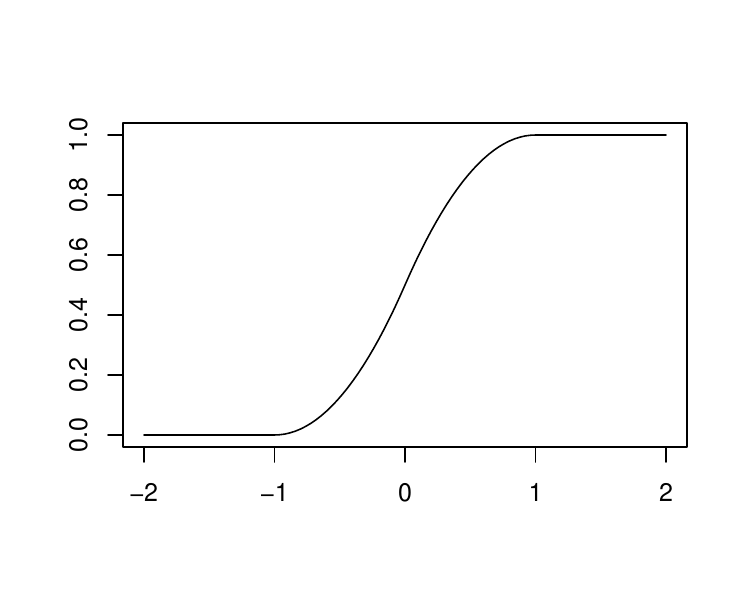}
		\caption{Plot of the smooth surrogate loss function with $\delta = 1$}
		\label{fig:surrogate}
	\end{center}
\end{figure}
\noindent
%Let $\calH$ be a Reproducing Kernel Hilbert Space (RKHS),
By using the surrogate function, we can estimate the optimal IDR under $M_{0,\gamma}(d)$ via computing
\begin{equation}\label{dc obj}
\min_{ \alpha \in \mathbb{R}, f \in \calH} \frac{1}{n}\sum_{i = 1}^{n} \frac{S(A_i f(\Xb_i))}{\pi(A_i | \Xb_i)}(-\alpha + \frac{(\alpha - R_i)_+}{\gamma}) + \frac{\lambda_n}{2}||f||^2_{\calH},
\end{equation}
where $||f||^2_{\calH}$ is some convex penalty function on $f$. For example, if we consider $\calH$ as a Reproducing Kernel Hilbert Space (RKHS), then $||f||^2_{\calH}$ could be the square of the RKHS norm of $f$. The estimated IDR is given by $\hat{d}_{0, \gamma}(\Xb) = \sign(\hat{f}(\Xb))$. Note that Problem \eqref{dc obj} involves a non-convex and potentially non-smooth optimization problem. Recent developments in difference-of-convex (DC) optimization (\cite{pang2016computing}) motivate us to use DC programming to efficiently solve this problem.  Note that $S(u)$ can be expressed as a difference of convex differentiable functions: $S_1(u) - S_2(u)$, where
\vspace{-0.3in}
\begin{multicols}{2}
	 \begin{equation*}
	S_{1}(u) =
	\begin{cases}
	0 & \text{if $u \leq -\delta$} \\
	(1 + u/\delta)^2/2 & \text{if $-\delta < u \leq 0$} \\
	1/2 + u/\delta& \text{if $0 <u$}\\
	\end{cases},
	\end{equation*}\break
	\begin{equation*}
	S_2(u) =
	\begin{cases}
	0 & \text{if $u \leq 0$} \\
	(u/\delta)^2/2 & \text{if $0 < u \leq \delta$} \\
	u/\delta - 1 /2& \text{if $u > \delta$}\\
	\end{cases}.
	\end{equation*}
\end{multicols}
%\[
%S_{1}(u) =
%\begin{cases}
%0 & \text{if $u \leq -\delta$} \\
%(1 + u/\delta)^2 & \text{if $-\delta < u \leq 0$} \\
%2 + 2u/\delta - 1 & \text{if $0 <u$}\\
%\end{cases},
%\]
%and 
%\[
%S_2(u) =
%\begin{cases}
%0 & \text{if $u \leq 0$} \\
%(u/\delta)^2 & \text{if $0 < u \leq \delta$} \\
%2u/\delta - 1 & \text{if $u > \delta$}\\
%\end{cases}.
%\]
Define
\begin{equation}
	G^{(1)}(f, \alpha) := \frac{1}{n}\sum_{i = 1}^{n} \frac{S(A_i f(\Xb_i))}{\pi(A_i | \Xb_i)}(- (\alpha - \frac{(\alpha - R_i)_+}{\gamma})) + \frac{\lambda_n}{2}||f||^2_{\calH},
\end{equation}
and
\begin{equation}
G^{(1)}_j(f) = \frac{1}{n}\sum_{i = 1}^{n} \frac{S(A_i f(\Xb_i))}{\pi(A_i | \Xb_i)}(- R_j + \frac{(R_j - R_i)_+}{\gamma}).
\end{equation}
The following proposition gives us a way to express \eqref{dc obj} as a DC function.
\begin{proposition}\label{equi prop}
	The following two optimization problems have the same optimal value, $i.e.$,
	\begin{equation}\label{equi}
		\min_{ \alpha \in \mathbb{R}, f \in \calH} G^{(1)}(f, \alpha) = \min_{f \in \calH}\left\{
		\widetilde{G}^{(1)}(f)\right\},
	\end{equation}
	where $\widetilde{G}^{(1)}(f) := \min_{1 \leq j \leq n}\{G^{(1)}_j(f)  \} + \frac{\lambda}{2}||f||^2_{\calH}$. More importantly, the optimal solution sets of $f$ to both problems are the same.	
\end{proposition}
Based on Proposition \ref{equi prop}, instead of solving \eqref{dc obj}, we can equivalently solve the optimization problem on the right hand side of \eqref{equi}.  Let $c_{ij} = \frac{- (R_j - \frac{(R_j - R_i)_+}{\gamma})}{\pi(A_i | \Xb_i)}$ for $i = 1,\cdots, n$ and $j = 1, \cdots, n$, and note that $c_{ij}$ is not necessarily nonnegative. Recall that $S(u) = S_1(u) - S_2(u)$. We can further rewrite $G_j^{(1)}(f)$ as
\begin{equation*}
	\begin{aligned}
		G_j^{(1)}(f) &= \frac{1}{n} \sum_{i = 1}^{n}(\max(c_{ij}, 0)S_1(A_if(\Xb_i)) + \max(-c_{ij}, 0)S_2(A_if(\Xb_i))) \\
		& - \frac{1}{n} \sum_{i = 1}^{n}(\max(c_{ij}, 0)S_2(A_if(\Xb_i)) + \max(-c_{ij}, 0)S_1(A_if(\Xb_i))) := F_j(f) - H_j(f), 
	\end{aligned}
\end{equation*}
where both $F_j(f)$ and $H_j(f)$ are convex functions with respect to $f$ for $j = 1, \cdots, n$. Then we can further decompose
\begin{equation*}\label{dc obj 2}
\begin{aligned}
	& \widetilde{G}^{(1)}(f) =  \min_{1 \leq j \leq n}\{F_j(f) - H_j(f)  \} + \frac{\lambda}{2}||f||^2_{\calH}\\
	= & \sum_{i  = 1}^{n}F_j(f) - \max_{1 \leq j \leq n}\{H_j(f) + \sum_{k \neq j}^{n}F_k(f)  \} + \frac{\lambda}{2}||f||^2_{\calH} 
	:= & F(f) - \max_{1 \leq j \leq n} h_j(f) + \frac{\lambda}{2}||f||^2_{\calH},
\end{aligned}
\end{equation*}
as a DC function, where $h_j(f):= H_j(f) + \sum_{k \neq j}^{n}F_k(f)$. Note that $\widetilde{G}^{(1)}(f)$ is a potentially non-smooth function if there exits multiple $k$'s such that $h_k(f) = \max_{1 \leq j \leq n} h_j(f) $. As discussed by \cite{pang2016computing}, traditional DC programming may  converge to nonsense points. To address this issue, motivated by \cite{pang2016computing}, we define $\calM_\epsilon(f) := \{j \ | \ h_j(f) \geq  \max_{1 \leq k \leq n} h_k(f) - \epsilon\}$, i.e., ``$\epsilon$-argmax" index set, and use the following enhanced probabilistic DC algorithm summarized in Table \ref{alg:dc 1} below to solve Problem \eqref{dc obj}.  
\begin{algorithm}[H]
	\caption{Algorithm for solving \eqref{dc obj}}\label{alg:dc 1}
	\begin{algorithmic}[1]
		\State Given a fixed $\epsilon > 0$, let $f^{(v)}$ be the solution at the $v$ iteration for $v = 1, 2, \cdots$.
		\State Randomly select $j \in \calM_\epsilon(f^{(v)})$, and compute
		\begin{equation}\label{subproblem for dc1}
			f^{(v+1)} \in \argmin_{f \in \calH} \{F(f) - \frac{\partial h_j(f^{(v)})}{\partial f}(f - f^{(v)}) + \frac{\lambda}{2}||f||^2_{\calH}    \}.
		\end{equation}
%		\State Let $\hat{j} = \argmin_j \widetilde{G}^{(1)}(f_j^{(v)})$; Set $f^{(v+1)} = \hat{f}_{\hat{j}}^{(v)}$
		\State The algorithm stops when $|\widetilde{G}^{(1)}(f^{(v)}) - \widetilde{G}^{(1)}(f^{(v+1)})| < \kappa$, for some pre-specified positive constant $\kappa$.
	\end{algorithmic}
\end{algorithm}
The proof of the convergence to sharp stationary points by the above algorithm can be found in \cite{pang2016computing}. For the computation of the subproblem \eqref{subproblem for dc1}, efficient algorithms such as the quasi-Newton method can be used. Compared with the algorithm proposed in \cite{qicuiliupang19}, by using the surrogate loss, we avoid directly solving a discontinuous optimization problem. While \cite{qicuiliupang19} transformed the problem into a constraint optimization by using the epigraph representation, it can be less efficient  than our proposed algorithm as the number of constraints is proportional to the number of observations. In addition, the convergence to a shaper stationary point can be guaranteed by our algorithm. For a general treatment on non-convex and non-smooth optimization problems, we refer to \cite{cui2021modern}. %, compared with those in \cite{qicuiliupang19}. 
\section{Analysis of Statistical Convergence}
In this section, we discuss the statistical theory related to our estimation under $M_{0,\gamma}(d)$. For the ease of presentation, we assume that Assumptions \ref{consistency}-\ref{overlap} always hold. With some abuse of notation, we first define
\begin{equation}\label{value fun 1}
M_{0,\gamma}(d, \alpha) = \Eb[\frac{\mathbb{I}(A = d(\Xb))}{\pi(A | \Xb)}(\alpha- \frac{1}{\gamma}(\alpha- R)_+)],
\end{equation}
and denote the optimal solution of maximizing \eqref{value fun 1} as $d_{0, \gamma} = \sign(f_0)$ and $\alpha_0$. Since we use the surrogate loss function $S(u)$, we further define
\begin{equation*}\label{value fun 3}
M_{S}(f, \alpha) = \Eb[\frac{S(Af(\Xb))}{\pi(A | \Xb)}(\alpha- \frac{1}{\gamma}(\alpha - R)_+)]
\end{equation*}
as the surrogate value function. Our theoretical results are similar to the standard regret analysis such as \cite{zhou2017residual}. The additional challenge is to handle the ``nuisance" parameter $\alpha$, which requires additional efforts. For example,  Fisher consistency cannot be directly obtained by using the standard proof  such as that in \cite{zhou2017residual}.
\subsection{Fisher Consistency}
We first establish Fisher consistency of estimating optimal ITRs under $M_S(f, \alpha)$ to justify the use of the surrogate loss $S(u)$, compared with $M_{0,\gamma}(d, \alpha)$. The proof is different from the classical Fisher consistency in classification, which only involves one functional class of interest. Here we need to consider the effect of the surrogate function on estimating $\alpha$, which is thus more involved.
\begin{theorem}\label{thm: optimal IDR}
	For any measurable function $f$ and $\alpha$ and a given $0 < \gamma \leq 1$, if $(f_S^\ast, \alpha_S^\ast)$ maximizes $M_S(f, \alpha)$, then $(\sign(f_S^\ast), \alpha_S^\ast)$ maximizes $M_{0,\gamma}(d, \alpha)$.
\end{theorem}
Based on Theorem \ref{thm: optimal IDR}, instead of $M_{0,\gamma}(d, \alpha)$, we can target on $M_S(d, \alpha)$ equivalently.
\subsection{Excess Value Bound}
Based on Theorem \ref{thm: optimal IDR}, we can further justify the use of the surrogate function $S(u)$ by establishing the following excess value bound for the $0$-$1$ loss in $M_{0,\gamma}(d, \alpha)$.
\begin{theorem}\label{excess value}
	For any measurable function $f, \alpha$ and any probability distribution over $(\Xb, A, R)$,
	\begin{equation*}
	M_{0,\gamma}(d_{0, \gamma}, \alpha_0) - M_{0,\gamma}(\sign(f), \alpha) \leq 2(M_S(f_S^\ast, \alpha_S^\ast) - M_S(f, \alpha)).
	\end{equation*}
\end{theorem}
Theorem \ref{excess value} gives us a way of bounding the difference between the optimal IDR and the estimated IDR under $M_{0,\gamma}(d, \alpha)$ by using $M_S(d, \alpha)$ instead.
\subsection{Convergence Rate}
In order to obtain the finite sample performance of our estimated optimal IDR under $M_{0,\gamma}(d, \alpha)$, it is enough to focus on the difference of $M_S(d, \alpha)$ between the estimated optimal IDR and the optimal ITR based on Theorem \ref{excess value}. Define
\begin{equation}\label{emp optimal}
	(\hat{f}, \hat{\alpha}) = \argmin_{f \in \calH, \alpha \in \mathbb{R}}O_n(f, \alpha)+ \frac{\lambda_n}{2}||f||^2_{\calH}, 
\end{equation}
where $O_n(f, \alpha) := \frac{1}{n}\sum_{i = 1}^{n} \frac{S(A_i f(\Xb_i))}{\pi(A_i | \Xb_i)}( \frac{(\alpha - R_i)_+}{\gamma} - \alpha)$. For simplicity, in the following we consider $\calH$
be a RKHS. The results can be extended to other scenarios such as the class of linear functions with the $l_1$ penalty. 
Define $O_T(f, \alpha) = -M_S(f, \alpha)$ and let
$(f_{\lambda_n}, \alpha_{\lambda_{n}}) =\argmin_{f \in \calH, \alpha \in \mathbb{R}} O_T(f, \alpha)+ \frac{\lambda_n}{2}||f||^2_{\calH}.$
Then $\calA(\lambda_n) := O_T(f_{\lambda_n}, \alpha_{\lambda_{2n}}) +  \frac{\lambda_n}{2}||f_{\lambda_n}||^2_{\calH} - O_T(f^\ast_S, \alpha_S^\ast)$ is considered to be the approximation error. The following theorem gives us a finite sample upper bound of our estimated optimal IDR and the optimal IDR based on $M_{0,\gamma}(d, \alpha)$.
\begin{theorem}\label{thm key}
	For any distribution $P$ over $(\Xb, A, R)$ such that the reward is uniformly bounded, i.e., $|R| \leq C_0$ for some positive constant $C_0$, if $\calA(\lambda_n) \leq C_1\lambda_n^{w_1}$, where $w_1\in (0, 1]$, then with probability at least $1-\epsilon$,
	\begin{equation*}
	M_{0,\gamma}(d_{0, \gamma}, \alpha_0) - M_{0,\gamma}(\sign(\hat{f}), \hat{\alpha}) \leq C_2 n^{-\frac{w_1}{2w_1 + 1}},
	\end{equation*}
	for some positive constant $C_2$. 
\end{theorem}
The above theorem shows that the difference between our estimated IDR and the optimal IDR under $M_{0,\gamma}(d, \alpha)$ converges to $0$ in probability under some conditions. The derivation requires additional control on the $\hat \alpha$, which can be shown to be uniformly bounded due to the optimization property.  The upper bound assumption on the approximation error $\calA(\lambda_n)$ is analogous to those in the  statistical learning literature such as \cite{steinwart2007fast} to derive the convergence rate. The convergence rate is the same as those in \cite{zhao2012estimating} and \cite{zhou2017residual}. This is not surprising because the proposed $M_{0,\gamma}(d)$ can be roughly regarded as the truncated mean.
\section{Numerical Studies}
We conduct extensive simulation studies and  real data analysis to illustrate the performance of our proposed method. For all simulation settings, we consider binary randomized trials with equal probabilities of patients being assigned to each treatment group. The performance will be similar when the difference in terms of the treatment assigned probabilities is not too large. In the extremely unbalanced case, the variance of the empirical approximation \eqref{dc obj} could be large, which thus requires a large sample to achieve desired performance. For the ease of presentation, we only present results of estimated optimal IDRs under $M_{0,\gamma}(d)$ with the $l_2$ penalty on $f$ in  Problem \eqref{dc obj}, where $f$ is a linear function of $\Xb$. We denote this method as $l2$-$\text{CVaR}_\gamma$.  %Here ``$l1$" and ``$l2$" refer to the $l_1$ and $l_2$ penalties. ``GK" represents using Gaussian radial basis functions with bandwidth $\varsigma$ to learn the optimal IDR.

All tuning parameters are selected based on the $10$-fold-cross-validation procedure. We select the tuning parameter that maximizes the empirical average of $M_{0,\gamma}(d)$ on the validation set. We compare our methods with the following four methods:
(1) the $l_1$-PLS by \cite{qian2011performance} with the basis function $(1, \Xb, A, \Xb A)$; (2) the RWL by \cite{zhou2017residual} with the  linear kernel and the $l_2$ penalty. Note that RWL used the same truncated hinge loss as ours; (3) the  $\gamma$-quantile optimal treatment regime denoted by $\text{IPWEQ}_\gamma$ (\cite{wang2017quantile}); (4) the model-based $\gamma$-quantile regression for the optimal treatment regime denoted by $\text{QR}_\gamma$ (\cite{linn2017interactive,xiao2019robust}). All results are based on 100 replications.

Since we only consider the randomized design study where $A$ is either $1$ or $-1$ with equal probabilities, we do not compare those methods designed for observational studies such as (\cite{zhang2012robust}).
\subsection{A Motivating Example}
We use one toy example to  demonstrate the necessity of tail control. In particular, the categorical covariate gender $X$ is generated by the uniform distribution over $\{1, -1\}$, where $1$ and $-1$ denote male and female respectively. Based on the randomly assigned treatment, the corresponding outcome $R$ is generated by the following model:
\begin{equation*}
	R = \mathbb{I}(XA =1) \epsilon_1 + \mathbb{I}(XA =-1) \epsilon_2,
\end{equation*}
where $\epsilon_1 \sim \calN(-0.1, 1)$ and $\epsilon_2 \sim \calN(0, 0.5)$. The corresponding plot is given in Figure \ref{fig:toy}. We consider training data with the sample size $n = 200$ and independently generated test data of size $10000$. Based on test data, in Figure \ref{fig:toy2}, we plot box plots of five different outcome distributions if treatments follow  estimated IDRs by $l_1$-PLS, linear-RWL, $\text{IPWEQ}_{0.5}$, $\text{QR}_{0.5}$, and $l2$-$\text{CVaR}_{0.5}$, respectively. Based on these box plots, we can observe that since there is not much difference between these two treatments based on the expected outcome, the empirical mean of value functions resulted from these five methods are indistinguishable. However, besides improving the expected or median outcome for each individual, our methods also control the tails of individuals, thus prefer treatments with less variability. The resulting outcome distribution by our method is more stable, and thus has less variability than the other four methods. If we choose different $\gamma$ for $\text{IPWEQ}_\gamma$ and $\text{QR}_\gamma$, e.g., 0.25, their resulting outcome distributions will be similar to ours since they focus more on the tail like our proposed method.
\begin{figure}[H]
	\begin{center}
		\includegraphics[scale = 0.65]{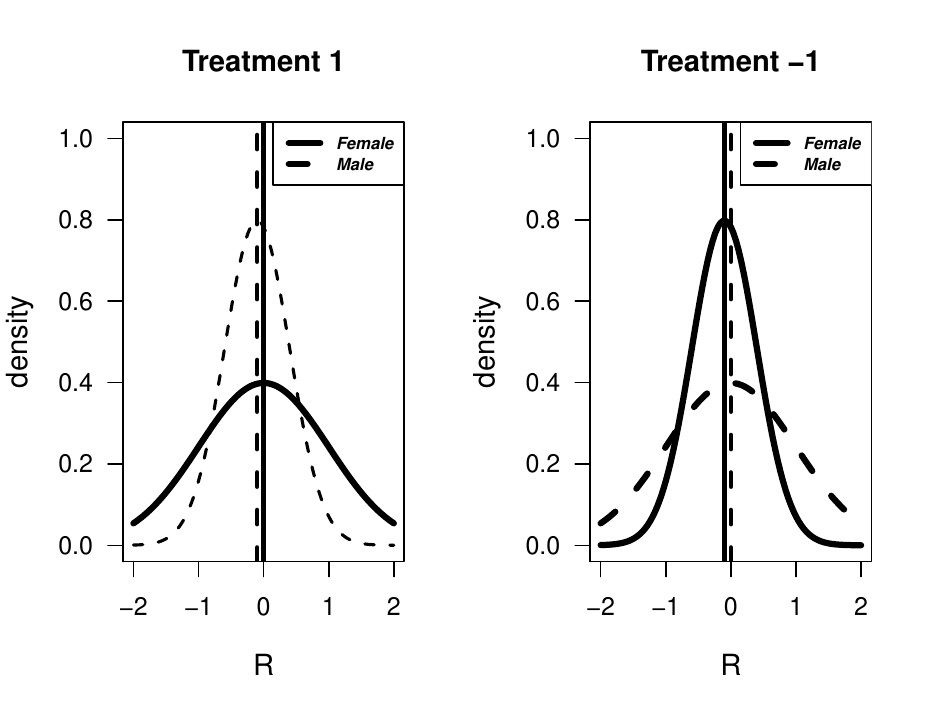}
		\caption{Plots of a motivating example. The dash and solid lines in the left plot show the probability densities of $\calN(-0.1, 0.5)$ and $\calN(0, 1)$ respectively. The dash and solid lines in the right plot correspond to the probability densities of $\calN(0, 1)$ and $\calN(-0.1, 0.5)$ respectively.}
		\label{fig:toy}
	\end{center}
\end{figure}

\begin{figure}[H]
	\begin{center}
		\includegraphics[scale = 0.5]{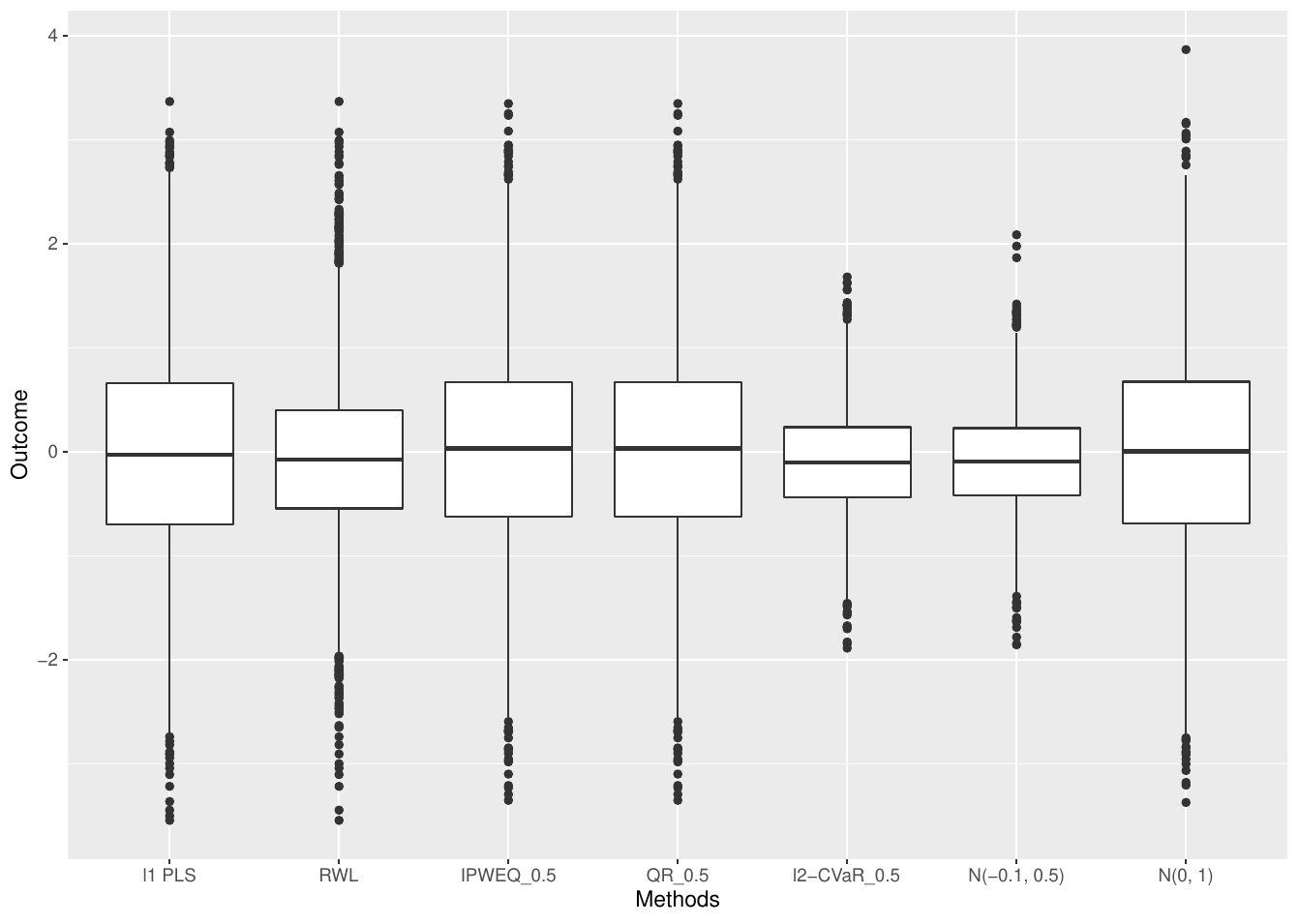}
		\caption{Box plots of value functions computed by five methods and two reference distributions. Methods starting from left to right are $l_1$-PLS, linear RWL, $\text{IPWEQ}_{0.5}$ $\text{QR}_{0.5}$, $l2$-$\text{CVaR}_{0.5}$ respectively. The last two box plots correspond to two reference normal distributions, i.e., $\calN(-0.1, 0.5)$ and $\calN(0, 1)$.}
		\label{fig:toy2}
	\end{center}
\end{figure}

%\begin{figure}[H]
%	\begin{center}
%		\includegraphics[scale = 0.4]{different_tau_gamma_value.eps}
%		\includegraphics[scale = 0.4]{different_tau_gamma_std.eps}
%		\caption{Medians and standard deviations of value functions under different combinations of $\tau$ and $\gamma$ by $l_2$-DC-CVaR. The left plot corresponds to the medians and the right plot corresponds to the standard deviations of value functions respectively.}
%		\label{fig:toy3}
%	\end{center}
%\end{figure}
\subsection{An illustrative Example}
To further illustrate how our proposed method differs from the mean and quantile optimal treatment regimes, we consider another simulating example, motivated by \cite{wang2017quantile}. Consider the reward $R = 1+ X + (3 - 5 X)A + (1 + A + 2AX)\epsilon$, where $\epsilon \sim \calN(0, 1)$ and $X \sim \text{uniform}(0, 1)$. We implement the above mentioned five methods with $\gamma = 0.1, 0.25, 0.5$ using $50000$ observations. Then we report their performance on $10^6$ testing observations under different criteria in Table \ref{fig: illustrative example}. As expected, methods perform the best under their corresponding criteria. For the $\text{QR}_\gamma$ method, while the goal is to find the best treatment with the largest conditional quantile treatment effect, it is not necessary to maximize the population quantile, i.e., $Q_\gamma(R(d))$. But their performances are also robust. In general, under the mean evaluation criterion, since our proposed methods focus on outcomes lower than certain quantiles, the performances are not as good as those mean-optimal treatment regimes, especially when $\gamma$ is small. However, when the evaluating criterion is on the low tails of outcomes, our methods are the most robust among all compared methods as $M_{0,\gamma}(d) \leq \min\{V(d), Q_\gamma(R(d))\}$. Furthermore, based on the reported coefficients of all estimated IDRs, it seems that robust IDRs given by our methods (or quantile based methods) tend to assign the treatment $A=-1$ as  the majority treatment as it  is better in improving lower tails of the outcome. Here we normalize the coefficients of IDRs so that all intercepts equal to $1$.
\begin{table}[H]
	\caption{Mean, different quantiles, and CVaR of the outcomes under 11 different treatment regimes given by different methods. The last two rows are the corresponding optimal IDRs' coefficients.
		\label{fig: illustrative example}}
	\centering
	\scalebox{0.7}{
	%	\begin{tabular}{cccccccccccc}
	\begin{tabular}{crrrrrrrrrrr}
			\hline
		Criterion & $l_1$PLS &  RWL & $\text{IPWEQ}_{0.5}$ & $\text{IPWEQ}_{0.25}$& $\text{IPWEQ}_{0.1}$&
			$\text{QR}_{0.5}$&
			$\text{QR}_{0.25}$&
			$\text{QR}_{0.1}$&
			$l2$$\text{CVaR}_{0.5}$&
			 $l2$$\text{CVaR}_{0.25}$ & $l2$$\text{CVaR}_{0.1}$\\ 
			\hline
			%				Dlearn&0.33(0.1)&0.5(0.01)&0.49(0.03)&0.48(0.07)&0.4(0.06)&0.51(0.13)&0.48(0.13)&0.48(0.13)\\
			Mean &2.802&\textbf{2.803}&2.747&2.798&2.721&2.803&2.707&2.46&2.761&2.684&2.499\\
			%				RWL-GK&0.5(0.01)&0.5(0.01)&0.5(0.01)&0.5(0.01)&\textbf{0.38}(0.04)&0.49(0.12)&0.43(0.11)&0.44(0.11)\\
			$Q_{0.5}$&2.81&2.801&\textbf{2.853}&2.761&2.594&2.797&2.569&2.233&2.666&2.532&2.281\\
			$Q_{0.25}$&1.32&1.325&1.191&\textbf{1.331}&1.246&1.326&1.226&0.891&1.296&1.197&0.943\\
			$Q_{0.1}$ &-0.153&-0.116&-0.566&-0.018&\textbf{0.119}&-0.106&0.119&-0.037&0.094&0.118&-0.002\\
			$\text{CVaR}_{0.5}$ &0.932&0.947&0.727&0.985&0.997&0.952&0.988&0.766&\textbf{1.011}&0.974&0.806 \\
				$\text{CVaR}_{0.25}$ &-0.231&-0.197&-0.614&-0.099&0.066&-0.187&0.071&-0.024&0.026&\textbf{0.079}&0.004 \\
					$\text{CVaR}_{0.1}$ &-1.594&-1.535&-2.179&-1.355&-0.944&-1.516&-0.912&-0.744&-1.076&-0.864&\textbf{-0.744} \\
					\hline
					Opt. IDR & $l_1$-PLS &  RWL & $\text{IPWEQ}_{0.5}$ & $\text{IPWEQ}_{0.25}$& $\text{IPWEQ}_{0.1}$ &
					$\text{QR}_{0.5}$&
					$\text{QR}_{0.25}$&
					$\text{QR}_{0.1}$&
					$\text{CVaR}_{0.5}$ &
					$\text{CVaR}_{0.25}$ & $\text{CVaR}_{0.1}$\\
					\hline
					intercept & 1 & 1 & 1 & 1 & 1 & 1 & 1 & 1 & 1 & 1 & 1 \\
					coefficient & -1.635 & -1.663 & -1.420 &  -1.759 & -2.115 & -1.672 & -2.165 & 2.959 & -1.965 & -2.241 & -2.831 \\
					\hline

		\end{tabular}
	}
\end{table} 
\subsection{Distributional Shift Examples}
Next we demonstrate the superior performance of our methods under distributional shifts of both covariates $\Xb$ and outcome $R$ based on the dual representations of $M_{0,\gamma}(d)$ in \eqref{min-max M_1}. Since \eqref{min-max M_1} is related to the expected outcome under some class of distributions, we only compare with mean-optimal IDRs such as $l_1$-PLS and RWL. For the purpose of illustration,  we only consider the sample size $n = 200$ and the dimension $p=20$. The outcome $R$ is generated by the model: $R = 1 + X_1 + X_2 + A (X_1 - X_2 + X_3) + \epsilon$. We consider the following two distributional shift scenarios:
\begin{itemize}
	\item[(1)] Each training covariate follows a two-component Gaussian mixture distribution of $\calN(0, 1)$ and asymmetric log-normal distribution $lognorm(0, 2)$ with probabilities of the two mixture components to be $0.7$ and $0.3$ respectively and $\epsilon$ follows the standard Gaussian distribution. Each test covariate and the test $\epsilon$ follow the standard Gaussian distribution.
	%	\item[(2)] $m(\Xb) = 1 + x_1^2 + x_2^2  \ \ \delta(\Xb)=1.8(0.3 - x_1 - x_2)$ and $\epsilon$ follows standard normal distribution
	\item[(2)] Covariates $X$ are generated by the uniform distribution between $-1$ and $1$ and $\epsilon$ follows a two-component mixture distribution of $\calN(0, 1)$ and asymmetric log-normal distribution $lognorm(0, 2)$ with probabilities of the two mixture component to be $0.7$ and $0.3$ respectively. For the test data, covariates X follow uniform distribution from $-1$ to $1$ and the test $\epsilon$ follows the standard Gaussian distribution.
\end{itemize}
The first scenario considers the covariate distributional shift and the second scenario considers the outcome distributional shift. For simplicity, we only report misclassification error rates given by $l_1$-PLS, RWL and $l2$-$\text{CVaR}_{0.5}$ in Table \ref{fig: shift1-2}. For Scenario (1), since $l_1$-PLS assumes a linear model, its performance is not affected by the distributional shift of covariates. In contrast, RWL, which is based on maximizing the value function, depends heavily on correct approximations to the value function empirically. Thus the performance of RWL is worse than $l_1$-PLS under this scenario. For the estimated optimal IDR under our proposed $M_{0,\gamma}(d)$, the performance is superior to RWL because $M_{0,\gamma}(d)$ considers the perturbation of the covariate distributional shift. For Scenario (2), since the estimated optimal IDR under $M_{0,\gamma}(d)$ is a minimax estimator under the outcome distributional shift, the performance is much better than the other two methods developed under the value function framework.
\begin{table}[H]
	\caption{Comparisons of misclassification error rates (standard deviations) for  simulated examples with $n = 200$ and $p= 20$.
		\label{fig: shift1-2}}
	\centering
	\scalebox{0.8}{
		\begin{tabular}{rrrr}
			\hline
			& Scenario (1) &  Scenario (2)\\ 
			\hline
			%				Dlearn&0.33(0.1)&0.5(0.01)&0.49(0.03)&0.48(0.07)&0.4(0.06)&0.51(0.13)&0.48(0.13)&0.48(0.13)\\
				$l_1$-PLS&\textbf{0.07} (0.01)&0.41 (0.007)\\
		RWL&0.14 (0.006)&0.40 (0.008)\\
			%				RWL-GK&0.5(0.01)&0.5(0.01)&0.5(0.01)&0.5(0.01)&\textbf{0.38}(0.04)&0.49(0.12)&0.43(0.11)&0.44(0.11)\\
			$l2$-$\text{CVaR}_{0.5}$&0.08 (0.01) &\textbf{0.19} (0.004)\\
			%				l1-DC-CVaR&0.33(0.06)&\textbf{0.39}(0.08)&\textbf{0.36}(0.07)&\textbf{0.32}(0.07)&0.44(0.04)&\textbf{0.45}(0.04)&0.46(0.04)&0.44(0.04)\\
			%				GK-DC-CVaR&0.39(0.09)&0.5(0.04)&0.39(0.1)&0.38(0.12)&0.39(0.04)&0.5(0.07)&\textbf{0.42}(0.07)&\textbf{0.4}(0.06)\\
%			$l_2$-MM-CVaR&0.23(0.009)&0.31(0.01)\\
			%				l1-MM-CVaR&0.35(0.08)&0.45(0.08)&0.38(0.09)&0.37(0.11)&0.41(0.05)&0.48(0.07)&0.44(0.07)&0.44(0.07)\\
			%				GK-MM-CVaR&0.39(0.08)&0.4(0.09)&0.36(0.09)&0.35(0.11)&0.39(0.03)&0.47(0.09)&\textbf{0.42}(0.06)&0.41(0.07)\\
			\hline
		\end{tabular}
	}
\end{table} 
\subsection{ Additional Simulation Scenarios}
We further study the performance of our proposed methods via three additional simulation examples, where the treatment-covariate interactions are linear. Nonlinear scenarios can be found in the supplementary material. We consider different combinations of $n$ and $p$, i.e.,  $n = 500, 1000$ and $p = 5, 10, 20$, respectively. For the ease of presentation, we only present results of $n=1000$ and $p=10$ since it is close to our real data size. Other results can be found in the Supplementary material.  We consider covariates $\Xb$ generated by the uniform distribution between $-1$ and $1$. The outcome $R$ is generated by the model: $R = m(\Xb) + A \delta(\Xb) + \epsilon$. We consider the following three different combinations of $m(\Xb)$, $\delta(\Xb)$ and $\epsilon$:
\begin{itemize}
	\item[(1)] $m(\Xb) = 1 + X_1 + X_2$, $\delta(\Xb)= X_1 - X_2 + X_3$, and $\epsilon \sim \calN(0, 2)$;
	%	\item[(2)] $m(\Xb) = 1 + x_1^2 + x_2^2  \ \ \delta(\Xb)=1.8(0.3 - x_1 - x_2)$ and $\epsilon$ follows standard normal distribution
	\item[(2)] $m(\Xb) = 1 + X_1 + X_2$, $\delta(\Xb)= X_1 - X_2 + X_3$, and $\log(\epsilon) \sim \calN(0, 1)$;
	\item[(3)] $m(\Xb) = 1 + X_1 + X_2 + X_3 + X_4$, $\delta(\Xb)= \exp(2X_1)( -1 + X_1 - X_2 + X_3)$, and $\log(\epsilon) \sim \calN(0, 2)$.
\end{itemize}
In the first scenario, we consider the standard normal error distribution. This setting is oracle for $l_1$-PLS and $\text{QR}_{0.5}$. The second scenario considers the error generated by log-normal distribution, which has a heavy right tail. This setting is used to test the robustness of different IDRs and is oracle for $\text{QR}_\gamma$. The last setting is used to show the advantage of direct methods over model-based methods when the model is misspecified but the decision function is correctly specified. In addition, this scenario considers the same error distribution as the second one.  Note that in all these three scenarios, the treatment only interacts with $\delta(\Xb)$ in the mean model. Therefore the global optimal decision rule is always $\sign(\delta(\Xb))$.

For our proposed method, $\text{IPWEQ}_\gamma$ and $\text{QR}_\gamma$, three different $\gamma$'s, i.e., $\gamma = 0.1, 0.25, 0.5$. are considered. In order to evaluate different methods, we generate test data and compute mean, different quantiles and CVaRs on the test data under all estimated IDRs. More results can be found in Supplementary material. Since the global optimal IDR is known, i.e., $\sign(\delta(\Xb))$, we can compute misclassification error rates for different methods on the test data, denoted by overall-error-rate in all tables below. In addition, to further show the robustness of different IDRs, we also compute the upper-error-rate and lower-error-rate. The upper-error-rate refers to misclassification  error of observations with outcomes larger than $90\%$ quantiles and the lower-error-rate is for observations with outcomes smaller than $10\%$ quantiles. The results are summarized in Tables \ref{fig: s3n1000p10}-\ref{fig: s7n1000p10}. Overall, our methods show competitive performances among all methods. In particular, for Scenario (1), which is the standard simulation setting in the literature, our proposed methods performs well in finding optimal IDRs, despite the slightly worse performance than mean-optimal IDRs. Since we focus more on the lower-tail, this is not surprising. For Scenario (2), as the error distribution is heavily right-tailed, methods under the value function framework  ignore individuals with potentially high risks, i.e., outcome lower than the threshold, while only focusing on maximizing the value function. Therefore the resulting performances could be worse than robust methods. Among three robust methods, $\text{QR}_\gamma$ performs the best as it is oracle in this scenario. It can also be  seen that our method performs better than that of \cite{wang2017quantile} in improving outcomes of lower tails. One potential reason is due to the optimization instability in training $\text{IPWEQ}_\gamma$. In terms of misclassification errors, all robust methods have smaller lower-error-rates than mean-optimal methods, demonstrating the superior performance of these estimated IDRs on the individuals with poor outcomes. In addition, we can observe that our method is more accurate than $\text{IPWEQ}_\gamma$ in identifying optimal treatments for individuals with small outcomes. For Scenario (3), our method performs the best among all methods. This is not surprising as our method does not model the outcome directly and thus does not suffer from model misspecification like $\text{QR}_\gamma$. In addition, our method is more robust than those under the  value function framework, thus shows better empirical performance in this setting. Note that $\text{IPWEQ}_\gamma$ performs similar to ours as it is also a direct method. 
\begin{table}%[HTB]
	\caption{Setting 1 with $n= 1000$ and $p = 10$. Mean, different quantiles, and CVaR of the outcomes under 11 different treatment regimes given by different methods are reported. The corresponding standard errors are in the parentheses.
		\label{fig: s3n1000p10}}
	\centering
	\scalebox{0.75}{
		\begin{tabular}{cccccccccccc}
			\hline
			 Criterion&$l_1$-PLS&RWL&$\text{IPWEQ}_{0.5}$&$\text{IPWEQ}_{0.25}$&$\text{IPWEQ}_{0.10}$\\
			 mean&2.162(0.002)&2.152(0.003)&2.081(0.011)&2.108(0.009)&2.11(0.008)\\
			 $Q_{0.5}$&\textbf{2.151}(0.001)&2.141(0.002)&2.089(0.009)&2.107(0.007)&2.107(0.007)\\
			 %$Q_{0.25}$&1.101(0.003)&1.086(0.004)&0.984(0.016)&1.022(0.014)&1.025(0.011)&\textbf{1.103}(0.001)&1.102(0.001)&1.098(0.001)&1.073(0.007)&1.069(0.007)&1.054(0.011)\\
			 $Q_{0.1}$&0.192(0.006)&0.175(0.005)&0.016(0.026)&0.079(0.021)&0.088(0.016)\\
			 $\text{CVaR}_{0.5}$&0.931(0.004)&0.916(0.004)&0.798(0.019)&0.843(0.015)&0.849(0.012)\\
			 %$\text{CVaR}_{0.25}$&0.219(0.006)&0.202(0.005)&0.04(0.027)&0.104(0.021)&0.115(0.015)&\textbf{0.226}(0.001)&0.224(0.002)&0.218(0.002)&0.182(0.01)&0.177(0.009)&0.16(0.014)\\
			 $\text{CVaR}_{0.1}$&-0.482(0.008)&-0.5(0.007)&-0.718(0.038)&-0.63(0.028)&-0.611(0.02)\\
			 overall-error-rate&0.029(0.005)&0.05(0.004)&0.117(0.008)&0.097(0.008)&0.095(0.006)\\
			 upper-error-rate&\textbf{0.006}(0.001)&0.014(0.002)&0.032(0.003)&0.027(0.003)&0.029(0.003)\\
			 lower-error-rate&0.01(0.002)&0.013(0.001)&0.049(0.007)&0.034(0.005)&0.03(0.003)\\
			\hline
			\hline
			\hline
			Criterion&$\text{QR}_{0.5}$&$\text{QR}_{0.25}$&$\text{QR}_{0.1}$&$l2$-$\text{CVaR}_{0.5}$&$l2$-$\text{CVaR}_{0.25}$&$l2$-$\text{CVaR}_{0.1}$\\
			mean&\textbf{2.163}(0.001)&2.163(0.001)&2.16(0.001)&2.143(0.004)&2.141(0.004)&2.131(0.007)\\
			$Q_{0.5}$&2.151(0.001)&2.15(0.001)&2.148(0.001)&2.135(0.004)&2.132(0.004)&2.124(0.006)\\
			%$Q_{0.25}$&1.101(0.003)&1.086(0.004)&0.984(0.016)&1.022(0.014)&1.025(0.011)&\textbf{1.103}(0.001)&1.102(0.001)&1.098(0.001)&1.073(0.007)&1.069(0.007)&1.054(0.011)\\
			$Q_{0.1}$&\textbf{0.199}(0.001)&0.197(0.002)&0.191(0.002)&0.156(0.01)&0.15(0.009)&0.132(0.014)\\$\text{CVaR}_{0.5}$&0.935(0.001)&\textbf{0.933}(0.001)&0.929(0.002)&0.902(0.007)&0.897(0.007)&0.883(0.011)\\
			%$\text{CVaR}_{0.25}$&0.219(0.006)&0.202(0.005)&0.04(0.027)&0.104(0.021)&0.115(0.015)&\textbf{0.226}(0.001)&0.224(0.002)&0.218(0.002)&0.182(0.01)&0.177(0.009)&0.16(0.014)\\
			$\text{CVaR}_{0.1}$&\textbf{-0.47}(0.002)&-0.473(0.002)&-0.481(0.004)&-0.527(0.013)&-0.532(0.011)&-0.551(0.016)\\
			overall-error-rate&\textbf{0.026}(0.002)&0.028(0.002)&0.035(0.002)&0.061(0.006)&0.065(0.005)&0.075(0.006)\\upper-error-rate&0.007(0.001)&0.008(0.001)&0.009(0.001)&0.016(0.002)&0.017(0.002)&0.022(0.003)\\lower-error-rate&\textbf{0.007}(0.001)&0.008(0.001)&0.01(0.001)&0.018(0.002)&0.019(0.002)&0.02(0.002)\\
			\hline
		\end{tabular}
	}
\end{table}
\begin{table}%[H]
	\caption{Setting 2 with $n= 1000$ and $p = 10$. Mean, different quantiles, and CVaR of the outcomes under 11 different treatment regimes given by different methods are reported. The corresponding standard errors are in the parentheses.
		\label{fig: s4n1000p10}}
	\centering
	\scalebox{0.75}{
		\begin{tabular}{cccccccccccc}
			\hline
			 Criterion&$l_1$-PLS&RWL&$\text{IPWEQ}_{0.5}$&$\text{IPWEQ}_{0.25}$&$\text{IPWEQ}_{0.10}$\\
			 mean&8.875(0.122)&8.933(0.106)&9.316(0.029)&9.42(0.017)&9.459(0.01)\\$Q_{0.5}$&3.208(0.107)&3.268(0.092)&3.574(0.011)&3.62(0.006)&3.635(0.004)\\
			 %$Q_{0.25}$&1.61(0.131)&1.675(0.116)&2.111(0.025)&2.219(0.015)&2.259(0.01)\\$Q_{0.1}$&0.415(0.163)&0.49(0.145)&1.025(0.051)&1.214(0.027)&1.292(0.017)\\
			 $\text{CVaR}_{0.5}$&1.454(0.136)&1.514(0.119)&1.961(0.034)&2.094(0.02)&2.148(0.012)\\
			 %$\text{CVaR}_{0.25}$&0.499(0.159)&0.556(0.137)&1.075(0.053)&1.271(0.031)&1.356(0.018)\\
			 $\text{CVaR}_{0.1}$&-0.337(0.187)&-0.301(0.157)&0.281(0.084)&0.559(0.049)&0.686(0.027)\\overall-error-rate&0.336(0.039)&0.318(0.034)&0.18(0.012)&0.127(0.009)&0.099(0.007)\\upper-error-rate&0.334(0.039)&0.317(0.034)&0.177(0.011)&0.125(0.009)&0.098(0.007)\\lower-error-rate&0.265(0.06)&0.247(0.052)&0.096(0.019)&0.047(0.01)&0.027(0.005)\\
			\hline
			\hline
			\hline
			Criterion&$\text{QR}_{0.5}$&$\text{QR}_{0.25}$&$\text{QR}_{0.1}$&$l2$-$\text{CVaR}_{0.5}$&$l2$-$\text{CVaR}_{0.25}$&$l2$-$\text{CVaR}_{0.1}$\\mean&9.504(0.006)&\textbf{9.512}(0.003)&9.507(0.002)&9.485(0.008)&9.502(0.006)&9.501(0.006)\\$Q_{0.5}$&3.656(0.001)&3.662(0)&\textbf{3.663}(0)&3.647(0.003)&3.652(0.002)&3.651(0.002)\\
			%$Q_{0.25}$&2.308(0.003)&2.326(0)&\textbf{2.327}(0)&2.286(0.006)&2.301(0.004)&2.299(0.004)\\
			$Q_{0.1}$&1.405(0.001)&\textbf{1.408}(0)&1.335(0.011)&1.363(0.007)&1.364(0.007)&1.364(0.007)\\$\text{CVaR}_{0.5}$&2.21(0.003)&2.231(0)&\textbf{2.233}(0)&2.183(0.007)&2.201(0.005)&2.2(0.005)\\
			%$\text{CVaR}_{0.25}$&1.448(0.005)&1.479(0.001)&\textbf{1.481}(0)&1.407(0.011)&1.433(0.007)&1.434(0.007)\\
			$\text{CVaR}_{0.1}$&0.818(0.008)&0.861(0.001)&\textbf{0.865}(0)&0.76(0.016)&0.797(0.011)&0.798(0.01)\\overall-error-rate&0.051(0.004)&0.014(0.001)&\textbf{0.006}(0)&0.077(0.006)&0.062(0.005)&0.063(0.005)\\upper-error-rate&0.051(0.004)&0.014(0.001)&\textbf{0.006}(0)&0.076(0.006)&0.061(0.005)&0.063(0.005)\\lower-error-rate&0.01(0.002)&0.003(0)&\textbf{0.001}(0)&0.015(0.003)&0.011(0.002)&0.011(0.002)\\
			\hline
		\end{tabular}
	}
\end{table}

\begin{table}%[H]
	\caption{Setting 3 with $n= 1000$ and $p = 10$. Mean, different quantiles, and CVaR of the outcomes under 11 different treatment regimes given by different methods are reported. The corresponding standard errors are in the parentheses.
		\label{fig: s7n1000p10}}
	\centering
	\scalebox{0.75}{
		\begin{tabular}{cccccccccccc}
			\hline
			 Criterion&$l_1$-PLS&RWL&$\text{IPWEQ}_{0.5}$&$\text{IPWEQ}_{0.25}$&$\text{IPWEQ}_{0.10}$\\mean&9.228(0.232)&9.302(0.208)&10.123(0.11)&10.351(0.034)&10.355(0.03)\\
			 $Q_{0.5}$&3.439(0.152)&3.511(0.132)&4.074(0.069)&4.227(0.023)&4.232(0.02)\\
			 %$Q_{0.25}$&1.457(0.151)&1.503(0.142)&2.074(0.081)&2.249(0.028)&2.258(0.023)&1.985(0.022)&1.991(0.016)&1.99(0.022)&2.327(0.017)&\textbf{2.335}(0.015)&2.328(0.013)\\
			 $Q_{0.1}$&-0.177(0.264)&-0.144(0.254)&0.738(0.141)&1.023(0.048)&1.044(0.037)\\
			 $\text{CVaR}_{0.5}$&0.924(0.239)&0.967(0.221)&1.845(0.128)&2.12(0.046)&2.138(0.037)\\
			 %$\text{CVaR}_{0.25}$&-0.573(0.336)&-0.547(0.315)&0.667(0.186)&1.059(0.067)&1.089(0.054)&0.47(0.051)&0.521(0.034)&0.473(0.053)&1.211(0.039)&\textbf{1.229}(0.035)&1.221(0.026)\\
			 $\text{CVaR}_{0.1}$&-2.555(0.564)&-2.544(0.522)&-0.519(0.312)&0.138(0.115)&0.191(0.094)\\overall-error-rate&0.348(0.041)&0.333(0.037)&0.168(0.026)&0.102(0.013)&0.098(0.011)\\upper-error-rate&0.34(0.042)&0.322(0.038)&0.155(0.024)&0.094(0.012)&0.09(0.01)\\lower-error-rate&0.28(0.055)&0.278(0.053)&0.09(0.031)&0.028(0.011)&0.023(0.009)\\
			\hline
			\hline
			\hline
			\hline
			Criterion &$\text{QR}_{0.5}$&$\text{QR}_{0.25}$&$\text{QR}_{0.1}$&$l2$-$\text{CVaR}_{0.5}$&$l2$-$\text{CVaR}_{0.25}$&$l2$-$\text{CVaR}_{0.1}$\\mean&10.01(0.032)&10.034(0.022)&10.014(0.033)&10.419(0.018)&\textbf{10.427}(0.017)&10.424(0.013)\\$Q_{0.5}$&4(0.02)&4.002(0.015)&4.004(0.021)&4.291(0.013)&\textbf{4.296}(0.011)&4.29(0.01)\\
			%$Q_{0.25}$&1.457(0.151)&1.503(0.142)&2.074(0.081)&2.249(0.028)&2.258(0.023)&1.985(0.022)&1.991(0.016)&1.99(0.022)&2.327(0.017)&\textbf{2.335}(0.015)&2.328(0.013)\\
			$Q_{0.1}$&0.602(0.034)&0.633(0.023)&0.605(0.036)&1.137(0.029)&\textbf{1.151}(0.026)&1.142(0.021)\\$\text{CVaR}_{0.5}$&1.707(0.035)&1.734(0.024)&1.711(0.036)&2.23(0.026)&2.242(0.024)&\textbf{2.235}(0.018)\\
			%$\text{CVaR}_{0.25}$&-0.573(0.336)&-0.547(0.315)&0.667(0.186)&1.059(0.067)&1.089(0.054)&0.47(0.051)&0.521(0.034)&0.473(0.053)&1.211(0.039)&\textbf{1.229}(0.035)&1.221(0.026)\\
			$\text{CVaR}_{0.1}$&-0.851(0.091)&-0.746(0.058)&-0.849(0.092)&0.379(0.067)&\textbf{0.408}(0.059)&0.4(0.041)\\overall-error-rate&0.208(0.007)&0.208(0.005)&0.207(0.007)&0.059(0.009)&\textbf{0.053}(0.008)&0.058(0.008)\\upper-error-rate&0.19(0.006)&0.191(0.005)&0.189(0.007)&0.054(0.008)&\textbf{0.049}(0.008)&0.054(0.007)\\lower-error-rate&0.12(0.009)&0.11(0.006)&0.12(0.009)&0.008(0.006)&\textbf{0.006}(0.005)&\textbf{0.006}(0.003)\\
			\hline
		\end{tabular}
	}
\end{table}

% latex table generated in R 3.3.2 by xtable 1.8-2 package
% Mon Oct 23 10:01:50 2017
% latex table generated in R 3.3.2 by xtable 1.8-2 package
% Mon Oct 23 10:02:39 2017
\subsection{Real Data Applications}
We perform a real data analysis to further evaluate our proposed robust criterion for estimating optimal IDRs. We use the clinical trial dataset from ``AIDS Clinical Trials Group (ACTG) 175" in \cite{hammer1996trial} to study whether there exists some subpopulations that are suitable for different combinations of treatments for AIDS. In this study, a total number of $2139$ patients with HIV infection is randomly assigned into four treatment groups: zidovudine (ZDV) monotherapy, ZDV combined with didanosine (ddI), ZDV combined with zalcitabine (ZAL), and ddI monotherapy with equal probabilities. In this data application, we focus on finding optimal IDRs between two treatments: ZDV with ddI and ZDV with ZAL as our interest. The total number of patients receiving these two treatments is 1046.\par
Similar to the previous studies by \cite{lu2013variable} and \cite{fan2016concordance}, we select $12$ baseline covariates for our model: age (year), weight(kg), CD4 T cells amount at baseline, Karnofsky score (scale at 0-100), CD8 amount at baseline, gender ($1$ = male, $0$ = female), homosexual activity ($1$ = yes, $0$ = no), race ($1$ = non white, $0$ = white),  history of intravenous drug use ($1$ = yes, $0$ = no), symptomatic status ($1$=symptomatic, $0$=asymptomatic), antiretroviral history ($1$=experienced, $0$=naive) and hemophilia ($1$=yes, $0$=no). The first five covariates are continuous and have been scaled before estimation. The remaining seven covariates are binary categorical variables. We consider the outcome as the difference between early stage (around $25$ weeks) CD4+ T (cells/mm$^3$) cell amount and  the baseline. Using this outcome, we can estimate the optimal IDR under our proposed robust criterion. To evaluate the performance of our proposed methods under the robust criterion, we randomly divide the dataset into five folds and use four of them to estimate optimal IDRs by different methods. The remaining one fold of data is used to evaluate the performances of different methods. We repeat this procedure $200$ times.  For each method, we report same quantities as before except misclassification errors as we do not know the truth. 
\begin{table}%[H]
	\caption{Results of value function comparison (standard errors) for the AIDS data. First column represents the means of empirical value functions. Second and third columns represent means of $50\%$ and $25\%$ quantiles of empirical value functions, respectively.
		\label{tab:data application}}
	\centering
	\scalebox{0.75}{
		\begin{tabular}{cccccccccccc}
			\hline
			 Criterion&$l_1$-PLS&RWL&$\text{IPWEQ}_{0.5}$&$\text{IPWEQ}_{0.25}$&$\text{IPWEQ}_{0.10}$\\mean&50.737(1.165)&48.679(1.2)&51.259(1.156)&49.47(1.215)&44.621(1.305)\\$Q_{0.5}$&40.57(1.267)&37.92(1.207)&42.085(1.241)&40.33(1.302)&35.565(1.305)\\$Q_{0.25}$&-27.665(1.21)&-28.387(1.196)&-26.222(1.104)&-24.613(1.257)&-32.305(1.312)\\$Q_{0.1}$&-99.447(2.226)&-97.038(2.318)&-100.372(2.222)&-95.981(2.115)&-107.521(2.364)\\$\text{CVaR}_{0.5}$&-50.56(1.194)&-50.598(1.317)&-50.403(1.193)&-48.507(1.323)&-54.844(1.398)\\$\text{CVaR}_{0.25}$&-108.622(1.894)&-107.12(2.012)&-109.023(1.903)&-104.985(2.031)&-112.784(1.996)\\$\text{CVaR}_{0.1}$&-186.459(3.329)&-183.688(3.446)&-189.305(3.34)&-180.473(3.443)&-188.242(3.279)\\
			\hline
			\hline
			\hline
			Criterion&$\text{QR}_{0.5}$&$\text{QR}_{0.25}$&$\text{QR}_{0.1}$&$l2$-$\text{CVaR}_{0.5}$&$l2$-$\text{CVaR}_{0.25}$&$l2$-$\text{CVaR}_{0.1}$\\mean&47.669(1.113)&46.268(1.169)&42.264(1.06)&53.345(1.341)&\textbf{55.042}(1.235)&48.711(1.818)\\$Q_{0.5}$&37.37(1.193)&35.31(1.252)&30.825(1.1)&44.855(1.433)&\textbf{46.85}(1.441)&41.07(1.88)\\$Q_{0.25}$&-29.47(1.216)&-31.63(1.088)&-32.682(1.126)&-25.552(1.295)&\textbf{-23.8}(1.261)&-28.613(1.531)\\$Q_{0.1}$&-100.229(2.198)&-101.004(2.196)&-98.603(1.998)&-100.084(2.228)&\textbf{-92.861}(2.335)&-96.221(2.282)\\$\text{CVaR}_{0.5}$&-52.678(1.214)&-53.882(1.203)&-51.421(1.129)&-47.621(1.403)&\textbf{-47.055}(1.32)&-50.579(1.448)\\$\text{CVaR}_{0.25}$&-109.928(1.903)&-110.554(1.894)&\textbf{-103.035}(1.658)&-106.173(2.08)&-105.262(2.023)&-107.446(1.899)\\$\text{CVaR}_{0.1}$&-189.68(3.263)&-189.613(3.58)&\textbf{-169.252}(2.902)&-181.626(3.556)&-186.562(3.611)&-185.402(3.379)\\
			\hline
		\end{tabular}
	}
\end{table}

From Table \ref{tab:data application}, we can see that our proposed methods are competitive. Due to the possible right tail distribution of $R$ as shown in Figure \ref{fig:real}, we observe the similar pattern as Scenario (2)  where the outcome distribution has potential right tails. Robust methods seem to outperform mean-based methods in this application as mean-based methods tend to ignore individuals with small outcomes due to the individuals with particularly large outcomes. Among all robust methods we compare, it seems that $l2$-$\text{CVaR}_{0.25}$ overall performs the best. It is somewhat surprising to see that while our methods in general have better performances in terms of quantiles and mean than other robust methods, $\text{QR}_\gamma$ methods show better performance than ours in terms of the CVaR criteria. This is possibly due to the difference between model-based and direct methods as observed in our simulation scenario (2). After comparison, we then apply our method on the whole dataset using $\gamma = 0.25$ and find that the estimated optimal rule tends to assign all individuals to the combined treatment: ZDV with ZAL, which indicates ZDV with ZAL has uniformly better treatment effect than ZDV with ddI under this criterion. This implies that if one considers personalized treatments for individuals at risk (below $25\%$ in this case), it may be more desirable to always assign the second treatment. In contrast, there seems to have a subgroup that can benefit from the treatment 1 according to the results from $l_1$-PLS and RWL. We also examine the estimated IDR using our method with $\gamma =0.5$. The most important variables that contribute to our estimated IDR are age, Karnofsky score, gender, race, antiretroviral history, and CD4 T cells at the baseline. In particular, our estimated IDR recommends ZDV with ZAL to old patients, but ZDV with ddI to young patients. This is consistent with the finding in the existing literature such as \cite{fan2016concordance}.   To further investigate whether the subgroup is valid or not under the proposed criterion, it will be interesting to develop and conduct a similar hypothesis testing under the CVaR criteria to that in \cite{shi2019sparse}. %, which we decide to leave it for future work.

%In particular, the ``GK-DC-CVaR" method performs the best compared with other methods, which indicates the optimal IDR of this problem may be potentially nonlinear. Another observation is that our proposed methods are not consistently better than other methods since robust methods are not necessarily the best for a specific application. However, robustness can be more insensitive to some deviations from model assumptions, which implies that our methods have the potential to be applied for a wide range of problems.
%\begin{table}[H]
%	\caption{Results of Value function comparison. First column represents the means of empirical value functions. Second and third columns represent means of $50\%$ and $25\%$ quantiles of empirical value functions, respectively.
%		\label{tab:data application}}
%	\centering
%	\begin{tabular}{rrrr}
%		\hline
%		& $V_n(d)$ & $50\%$ quantiles &$25\%$ quantiles \\ 
%		\hline
%		$l_1$-pls&53.02(12.55)&43.69(16.06)&$-$26.13(13.6)\\
%		RWL&53.74(12.2)&44.23(15.37)&-26.25(13.46)\\
%		RWL-GK&53.29(12.17)&43.59(12.93)&$-$25.59(12.93)\\
%		l2-DC-CVaR&54.65(13.26)&43.69(15.57)&-26.24(13.85)\\
%		l1-DC-CVaR&50.69(11.82)&38.22(13.74)&$-$29.96(12.55)\\
%		GK-DC-CVaR&\textbf{55.33}(11.85)&\textbf{45.91}(15.86)&\textbf{$-$23.02}(12.7)\\
%		\hline
%	\end{tabular}
%\end{table}

\begin{figure}%[H]
	\begin{center}
		\includegraphics[scale = 0.6]{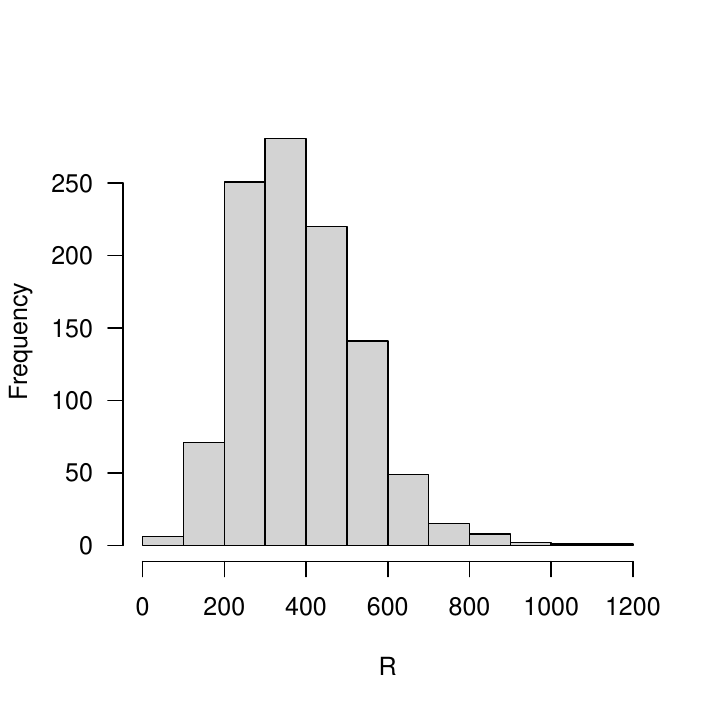}
		\caption{Histogram of the difference between early stage CD4+T cell and the baseline. It exhibits a  right tail behavior.}
		\label{fig:real}
	\end{center}
\end{figure}

\section{Discussion and Extensions}
In this paper, we discuss the robustness issue for ITR problems and propose new criteria to control tails of the individuals' outcomes. Our framework is broad and can cover a wide range of criteria. In this section, we further discuss two extensions of our proposed criteria in terms of  modelings. More discussions and results can be found in  the Supplementary Material.

\subsection{Control of Two-sided Truncated Outcome}
In this subsection, we extend the $\gamma$-CVaR criterion to  the truncated mean over a certain range of quantiles. Consider two quantiles $\gamma_1$ and $\gamma_2$ with $\gamma_1 > \gamma_2$, we evaluate each IDR $d$ by considering its truncated mean between $\gamma_2$ and $\gamma_1$ quantiles, i.e.,
\begin{align}\label{truncated mean}
 \Eb^d[R\mathbb{I}(R \leq Q_{\gamma_1}(R(d)))] - \Eb^d[R\mathbb{I}(R \leq Q_{\gamma_2}(R(d)))].
\end{align}
By the property of truncated mean, similar as $\gamma$-CVaR, we can use the following measure to evaluate each IDR $d$:
\begin{align}\label{truncated mean representation}
M_{\gamma_1, \gamma_2}(d) = \sup_{ \alpha_1 \in \mathbb{R}} \left\{\gamma_1\alpha_1 - \Eb[\frac{\mathbb{I}(A = d(\Xb))}{\pi(A | \Xb)}(\alpha_1 - R)_+] \right\} - \sup_{ \alpha_2 \in \mathbb{R}} \left\{\gamma_2\alpha_2 - \Eb[\frac{\mathbb{I}(A = d(\Xb))}{\pi(A | \Xb)}(\alpha_2 - R)_+] \right\},
\end{align}
which is the same as \eqref{truncated mean} when the outcome has an continuous distribution. Then the optimal IDR under this criterion is defined as
\begin{align}\label{optimal IDR under truancated mean}
	d_3 \in \argmax_{d \in \calD_0} M_{\gamma_1, \gamma_2}(d).
\end{align}
In order to estimate $d_3$, we can also replace the $0$-$1$ loss by the smooth surrogate loss function $S(u)$ and solve the following optimization problem.
\begin{equation}\label{empirical truncated}
\begin{aligned}
	\min_{f \in \calH}& \left\{\left(\min_{ \alpha_1 \in \mathbb{R}} \frac{1}{n}\sum_{i = 1}^{n} \frac{S(A_i f(\Xb_i))}{\pi(A_i | \Xb_i)}(-\gamma_1\alpha_1 + (\alpha_1 - R_i)_+)\right) \right.\\
	&\left. - \left(\min_{ \alpha_2 \in \mathbb{R}} \frac{1}{n}\sum_{i = 1}^{n} \frac{S(A_i f(\Xb_i))}{\pi(A_i | \Xb_i)}(-\gamma_2\alpha_2 + (\alpha_2 - R_i)_+)\right)\right.\\
	&\left.+ \frac{\lambda_n}{2}||f||^2_{\calH}\right\}.
\end{aligned}
\end{equation}
One can similarly use Proposition \ref{equi prop} to decompose the objective function in \eqref{empirical truncated} into a DC function and develop a corresponding DC algorithm.  In addition, we note that $d_3$ indeed can approximate the optimal quantile IDR $\widetilde d_\gamma$ when $\gamma_1$ and $\gamma_2$ are chosen to be close to $\gamma$. 
\subsection{Flexible Models}
In previous applications of our proposed method, we only use the CVaR criterion to estimate a robust IDR. In practice, one can combine it with other criteria. For example, in order to have tail controls and more mean-outcome improvement, we can consider to use $M_{1, \gamma}(d)$ defined below by combining our proposed $M_{0,\gamma}(d)$ and $V(d)$ together:
\begin{equation}\label{mix fun compute}
\begin{aligned}
M_{1, \gamma}(d) &:= 0.5V(d) + 0.5 M_{0,\gamma}(d).
\end{aligned}
\end{equation} 
All our results in Sections 3 and 4 can be naturally extended. This is indeed roughly equivalent to maximizing $V(d)$ among all the decision rules with $M_{0,\gamma}(d)$ larger than some threshold. Such an interpretation motivates us to consider another approach to find the optimal IDR when observing multiple outcomes after receiving treatments. Without loss of generality, suppose we can observe a risk outcome $H$ in addition to $R$. In general, we prefer a smaller risk outcome. By some modification on $M_{0,\gamma}(d)$, we can search the optimal IDRs by 
\begin{equation}\label{contraint CVaR}
	\begin{aligned}
	\max_{d \in \calD_0} \epc & V(d)\\
	\text{subject to} \epc & \min_{ \alpha \in \mathbb{R}} \left\{\alpha + \frac{1}{\gamma} \Eb[(H(d) - \alpha)_+]  \right\} \leq t, 
	\end{aligned}
\end{equation}
for some pre-specified constant $t$, so that we can control the risk outcome with the CVaR criterion. By a similar analysis as in Section 2, the resulting IDRs using \eqref{contraint CVaR} are the best IDRs among all the decision rules with the risk outcome of $(1-\gamma) \times 100\%$ of the population being less than some threshold $t$. Given the data, we could use the same techniques in Section 3 to estimate the IDR. Specifically, by using a surrogate function $S(u)$, we can solve the following optimization problem:
\begin{equation}\label{empirical contraint CVaR}
\begin{aligned}
\max_{d \in \calD} \epc & \frac{1}{n}\sum_{i  = 1}^n \frac{R_iS(A_i f(\Xb_i))}{\pi(A_i | \Xb_i)} - \frac{\lambda}{2} ||f||^2_\calH\\
\text{subject to} \epc & \min_{1 \leq j \leq n} \left\{ \frac{1}{n}\sum_{i  = 1}^n \frac{(H_j + \frac{1}{\gamma}(H_i - H_j)_+)S(A_i f(\Xb_i))}{\pi(A_i | \Xb_i)}  \right\} \leq t, 
\end{aligned}
\end{equation}
which  can be formulated as optimizing a DC function with a DC constraint as we can see that  the minimum of a finite number of DC functions in the constraint is still a DC function. However, several issues need to be solved before proceeding. The existence of the optimal solution for the optimization problem \eqref{empirical contraint CVaR} should be demonstrated. The use of a surrogate function to replace the indicator function in \eqref{contraint CVaR} needs to be further justified. Another challenge is to establish a regret bound for the optimal IDR obtained from Problem \eqref{empirical contraint CVaR}. Existing techniques in statistical learning theory may not be directly used since there is a stochastic term in the constraint of \eqref{contraint CVaR} and \eqref{empirical contraint CVaR}. Nevertheless, this can be an interesting direction to pursue in the future.

\begin{spacing}{0.9}
\bibliography{reference_v3}

\begin{thebibliography}{62}
\providecommand{\natexlab}[1]{#1}
\providecommand{\url}[1]{\texttt{#1}}
\expandafter\ifx\csname urlstyle\endcsname\relax
  \providecommand{\doi}[1]{doi: #1}\else
  \providecommand{\doi}{doi: \begingroup \urlstyle{rm}\Url}\fi

\bibitem[Artzner et~al.(1999)Artzner, Delbaen, Eber, and
  Heath]{artzner1999coherent}
P.~Artzner, F.~Delbaen, J.-M. Eber, and D.~Heath.
\newblock Coherent measures of risk.
\newblock \emph{Mathematical finance}, 9\penalty0 (3):\penalty0 203--228, 1999.

\bibitem[Athey and Wager(2021)]{athey2017efficient}
S.~Athey and S.~Wager.
\newblock Policy learning with observational data.
\newblock \emph{Econometrica}, 89\penalty0 (1):\penalty0 133--161, 2021.

\bibitem[Ben-Tal and Teboulle(1986)]{ben1986expected}
A.~Ben-Tal and M.~Teboulle.
\newblock Expected utility, penalty functions, and duality in stochastic
  nonlinear programming.
\newblock \emph{Management Science}, 32\penalty0 (11):\penalty0 1445--1466,
  1986.

\bibitem[Beygelzimer and Langford(2009)]{Beygelzimer:2009:OTL:1557019.1557040}
A.~Beygelzimer and J.~Langford.
\newblock The offset tree for learning with partial labels.
\newblock In \emph{Proceedings of the 15th ACM SIGKDD International Conference
  on Knowledge Discovery and Data Mining}, KDD '09, pages 129--138, New York,
  NY, USA, 2009. ACM.
\newblock ISBN 978-1-60558-495-9.
\newblock \doi{10.1145/1557019.1557040}.

\bibitem[Bhattacharya and Dupas(2012)]{bhattacharya2012inferring}
D.~Bhattacharya and P.~Dupas.
\newblock Inferring welfare maximizing treatment assignment under budget
  constraints.
\newblock \emph{Journal of Econometrics}, 167\penalty0 (1):\penalty0 168--196,
  2012.

\bibitem[Chamberlain(2011)]{chamberlain2011bayesian}
G.~Chamberlain.
\newblock Bayesian aspects of treatment choice.
\newblock In \emph{The Oxford Handbook of Bayesian Econometrics}. Oxford
  University Press, 2011.

\bibitem[Chen et~al.(2018)Chen, Fu, He, Kosorok, and Liu]{chen2018estimating}
J.~Chen, H.~Fu, X.~He, M.~R. Kosorok, and Y.~Liu.
\newblock Estimating individualized treatment rules for ordinal treatments.
\newblock \emph{Biometrics}, 74\penalty0 (3):\penalty0 924--933, 2018.

\bibitem[Chow et~al.(2017)Chow, Ghavamzadeh, Janson, and
  Pavone]{Chow:2017:RRL:3122009.3242024}
Y.~Chow, M.~Ghavamzadeh, L.~Janson, and M.~Pavone.
\newblock Risk-constrained reinforcement learning with percentile risk
  criteria.
\newblock \emph{Journal of Machine Learning Research}, 18\penalty0
  (1):\penalty0 6070--6120, Jan. 2017.
\newblock ISSN 1532-4435.

\bibitem[Cui and Pang(2021)]{cui2021modern}
Y.~Cui and J.-S. Pang.
\newblock Modern nonconvex nondifferentiable optimization, 2021.

\bibitem[Cui et~al.(2017)Cui, Zhu, and Kosorok]{cui2017tree}
Y.~Cui, R.~Zhu, and M.~Kosorok.
\newblock Tree based weighted learning for estimating individualized treatment
  rules with censored data.
\newblock \emph{Electronic Journal of Statistics}, 11\penalty0 (2):\penalty0
  3927--3953, 2017.

\bibitem[Dehejia(2008)]{dehejia2008ate}
R.~Dehejia.
\newblock When is {ATE} enough? {R}isk aversion and inequality aversion in
  evaluating training programs.
\newblock In \emph{Modelling and Evaluating Treatment Effects in Econometrics},
  pages 263--287. Emerald Group Publishing Limited, 2008.

\bibitem[Dud{\'\i}k et~al.(2011)Dud{\'\i}k, Langford, and Li]{dudik2011doubly}
M.~Dud{\'\i}k, J.~Langford, and L.~Li.
\newblock Doubly robust policy evaluation and learning.
\newblock \emph{arXiv preprint arXiv:1103.4601}, 2011.

\bibitem[Fan et~al.(2017)Fan, Lu, Song, and Zhou]{fan2016concordance}
C.~Fan, W.~Lu, R.~Song, and Y.~Zhou.
\newblock Concordance-assisted learning for estimating optimal individualized
  treatment regimes.
\newblock \emph{Journal of the Royal Statistical Society: Series B (Statistical
  Methodology)}, 79\penalty0 (5):\penalty0 1565--1582, 2017.

\bibitem[Fang et~al.(2021)Fang, Wang, and Wang]{fang2021fairness}
E.~X. Fang, Z.~Wang, and L.~Wang.
\newblock Fairness-oriented learning for optimal individualized treatment
  rules.
\newblock \emph{Journal of the American Statistical Association}, \penalty0
  (just-accepted):\penalty0 1--31, 2021.

\bibitem[Gunter et~al.(2011)Gunter, Zhu, and Murphy]{gunter2011variable}
L.~Gunter, J.~Zhu, and S.~Murphy.
\newblock Variable selection for qualitative interactions.
\newblock \emph{Statistical methodology}, 8\penalty0 (1):\penalty0 42--55,
  2011.

\bibitem[Hammer et~al.(1996)Hammer, Katzenstein, Hughes, Gundacker, Schooley,
  Haubrich, Henry, Lederman, Phair, Niu, et~al.]{hammer1996trial}
S.~M. Hammer, D.~A. Katzenstein, M.~D. Hughes, H.~Gundacker, R.~T. Schooley,
  R.~H. Haubrich, W.~K. Henry, M.~M. Lederman, J.~P. Phair, M.~Niu, et~al.
\newblock A trial comparing nucleoside monotherapy with combination therapy in
  hiv-infected adults with cd4 cell counts from 200 to 500 per cubic
  millimeter.
\newblock \emph{New England Journal of Medicine}, 335\penalty0 (15):\penalty0
  1081--1090, 1996.

\bibitem[Hirano and Porter(2009)]{hirano2009asymptotics}
K.~Hirano and J.~R. Porter.
\newblock Asymptotics for statistical treatment rules.
\newblock \emph{Econometrica}, 77\penalty0 (5):\penalty0 1683--1701, 2009.

\bibitem[Kallus(2018)]{NIPS2018_8105}
N.~Kallus.
\newblock Balanced policy evaluation and learning.
\newblock In \emph{Advances in Neural Information Processing Systems 31}, pages
  8895--8906. Curran Associates, Inc., 2018.

\bibitem[Kasy(2016)]{kasy2016partial}
M.~Kasy.
\newblock Partial identification, distributional preferences, and the welfare
  ranking of policies.
\newblock \emph{Review of Economics and Statistics}, 98\penalty0 (1):\penalty0
  111--131, 2016.

\bibitem[Kitagawa and Tetenov(2018)]{kitagawa2018should}
T.~Kitagawa and A.~Tetenov.
\newblock Who should be treated? empirical welfare maximization methods for
  treatment choice.
\newblock \emph{Econometrica}, 86\penalty0 (2):\penalty0 591--616, 2018.

\bibitem[Laber and Zhao(2015)]{laber2015tree}
E.~Laber and Y.~Zhao.
\newblock Tree-based methods for individualized treatment regimes.
\newblock \emph{Biometrika}, 102\penalty0 (3):\penalty0 501--514, 2015.

\bibitem[Linn et~al.(2017)Linn, Laber, and Stefanski]{linn2017interactive}
K.~A. Linn, E.~B. Laber, and L.~A. Stefanski.
\newblock Interactive {Q}-learning for quantiles.
\newblock \emph{Journal of the American Statistical Association}, 112\penalty0
  (518):\penalty0 638--649, 2017.

\bibitem[Liu et~al.(2018)Liu, Wang, Kosorok, Zhao, and Zeng]{liu2016robust}
Y.~Liu, Y.~Wang, M.~R. Kosorok, Y.~Zhao, and D.~Zeng.
\newblock Augmented outcome-weighted learning for estimating optimal dynamic
  treatment regimens.
\newblock \emph{Statistics in Medicine}, 37\penalty0 (26):\penalty0 3776--3788,
  2018.

\bibitem[Lu et~al.(2013)Lu, Zhang, and Zeng]{lu2013variable}
W.~Lu, H.~H. Zhang, and D.~Zeng.
\newblock Variable selection for optimal treatment decision.
\newblock \emph{Statistical methods in medical research}, 22\penalty0
  (5):\penalty0 493--504, 2013.

\bibitem[Manski(2004)]{manski2004statistical}
C.~F. Manski.
\newblock Statistical treatment rules for heterogeneous populations.
\newblock \emph{Econometrica}, 72\penalty0 (4):\penalty0 1221--1246, 2004.

\bibitem[Mo et~al.(2020)Mo, Qi, and Liu]{mo20}
W.~Mo, Z.~Qi, and Y.~Liu.
\newblock Learning optimal distributionally robust individualized treatment
  rules.
\newblock \emph{Journal of the American Statistical Association}, page to
  appear, 2020.

\bibitem[Murphy(2003)]{murphy2003optimal}
S.~A. Murphy.
\newblock Optimal dynamic treatment regimes.
\newblock \emph{Journal of the Royal Statistical Society: Series B (Statistical
  Methodology)}, 65\penalty0 (2):\penalty0 331--355, 2003.

\bibitem[Murphy(2005)]{murphy2005generalization}
S.~A. Murphy.
\newblock A generalization error for {Q}-learning.
\newblock \emph{Journal of Machine Learning Research}, 6\penalty0
  (Jul):\penalty0 1073--1097, 2005.

\bibitem[Pang et~al.(2016)Pang, Razaviyayn, and Alvarado]{pang2016computing}
J.-S. Pang, M.~Razaviyayn, and A.~Alvarado.
\newblock Computing {B}-stationary points of nonsmooth dc programs.
\newblock \emph{Mathematics of Operations Research}, 42\penalty0 (1):\penalty0
  95--118, 2016.

\bibitem[Pflug(2000)]{pflug2000some}
G.~C. Pflug.
\newblock Some remarks on the value-at-risk and the conditional value-at-risk.
\newblock In \emph{Probabilistic Constrained Optimization}, pages 272--281.
  Springer, 2000.

\bibitem[Qi and Liu(2018)]{qi2017}
Z.~Qi and Y.~Liu.
\newblock D-learning to estimate optimal individual treatment rules.
\newblock \emph{Electronic Journal of Statistics}, 12\penalty0 (2):\penalty0
  3601--3638, 2018.

\bibitem[Qi et~al.(2019{\natexlab{a}})Qi, Cui, Liu, and Pang]{qicuiliupang19}
Z.~Qi, Y.~Cui, Y.~Liu, and J.~Pang.
\newblock Estimation of individualized decision rules based on an optimized
  covariate-dependent equivalent of random outcomes.
\newblock \emph{SIAM Journal on Optimization}, 29\penalty0 (3):\penalty0
  2337--2362., 2019{\natexlab{a}}.

\bibitem[Qi et~al.(2019{\natexlab{b}})Qi, Liu, Fu, and Liu]{qi2018}
Z.~Qi, D.~Liu, H.~Fu, and Y.~Liu.
\newblock Multi-armed angle-based direct learning for estimating optimal
  individualized treatment rules with various outcomes.
\newblock \emph{Journal of the American Statistical Association},
  2019{\natexlab{b}}.

\bibitem[Qian and Murphy(2011)]{qian2011performance}
M.~Qian and S.~A. Murphy.
\newblock Performance guarantees for individualized treatment rules.
\newblock \emph{The Annals of Statistics}, 39\penalty0 (2):\penalty0 1180,
  2011.

\bibitem[Robins(2004)]{robins2004optimal}
J.~M. Robins.
\newblock Optimal structural nested models for optimal sequential decisions.
\newblock In \emph{Proceedings of the second seattle Symposium in
  Biostatistics}, pages 189--326. Springer, 2004.

\bibitem[Rockafellar(1974)]{rockafellar1974conjugate}
R.~Rockafellar.
\newblock \emph{Conjugate Duality and Optimization}.
\newblock Society for Industrial and Applied Mathematics, 1974.
\newblock \doi{10.1137/1.9781611970524}.
\newblock URL \url{https://epubs.siam.org/doi/abs/10.1137/1.9781611970524}.

\bibitem[Rockafellar and Royset(2010)]{rockafellar2010buffered}
R.~T. Rockafellar and J.~O. Royset.
\newblock On buffered failure probability in design and optimization of
  structures.
\newblock \emph{Reliability engineering \& system safety}, 95\penalty0
  (5):\penalty0 499--510, 2010.

\bibitem[Rockafellar and Uryasev(2000)]{rockafellar2000optimization}
R.~T. Rockafellar and S.~Uryasev.
\newblock Optimization of conditional value-at-risk.
\newblock \emph{Journal of Risk}, 2:\penalty0 21--42, 2000.

\bibitem[Rockafellar and Uryasev(2002)]{rockafellar2002conditional}
R.~T. Rockafellar and S.~Uryasev.
\newblock Conditional value-at-risk for general loss distributions.
\newblock \emph{Journal of Banking \& Finance}, 26\penalty0 (7):\penalty0
  1443--1471, 2002.

\bibitem[Rockafellar and Wets(2009)]{rockafellar2009variational}
R.~T. Rockafellar and R.~J.-B. Wets.
\newblock \emph{Variational analysis}, volume 317.
\newblock Springer Science \& Business Media, 2009.

\bibitem[Rubin(1974)]{rubin1974estimating}
D.~B. Rubin.
\newblock Estimating causal effects of treatments in randomized and
  nonrandomized studies.
\newblock \emph{Journal of educational Psychology}, 66\penalty0 (5):\penalty0
  688, 1974.

\bibitem[Sarykalin et~al.(2008)Sarykalin, Serraino, and
  Uryasev]{sarykalin2008value}
S.~Sarykalin, G.~Serraino, and S.~Uryasev.
\newblock Value-at-risk vs. conditional value-at-risk in risk management and
  optimization.
\newblock \emph{Tutorials in Operations Research}, pages 270--294, 2008.

\bibitem[Savage(1951)]{savage1951theory}
L.~J. Savage.
\newblock The theory of statistical decision.
\newblock \emph{Journal of the American Statistical association}, 46\penalty0
  (253):\penalty0 55--67, 1951.

\bibitem[Schulte et~al.(2014)Schulte, Tsiatis, Laber, and
  Davidian]{schulte2014q}
P.~J. Schulte, A.~A. Tsiatis, E.~B. Laber, and M.~Davidian.
\newblock Q-and {A}-learning methods for estimating optimal dynamic treatment
  regimes.
\newblock \emph{Statistical science}, 29\penalty0 (4):\penalty0 640, 2014.

\bibitem[Shi et~al.(2018)Shi, Fan, Song, and Lu]{shi2018high}
C.~Shi, A.~Fan, R.~Song, and W.~Lu.
\newblock High-dimensional a-learning for optimal dynamic treatment regimes.
\newblock \emph{Annals of statistics}, 46\penalty0 (3):\penalty0 925, 2018.

\bibitem[Shi et~al.(2019)Shi, Lu, and Song]{shi2019sparse}
C.~Shi, W.~Lu, and R.~Song.
\newblock A sparse random projection-based test for overall qualitative
  treatment effects.
\newblock \emph{Journal of the American Statistical Association}, 2019.

\bibitem[Steinwart and Scovel(2007)]{steinwart2007fast}
I.~Steinwart and C.~Scovel.
\newblock Fast rates for support vector machines using gaussian kernels.
\newblock \emph{The Annals of Statistics}, pages 575--607, 2007.

\bibitem[Stoye(2009)]{stoye2009minimax}
J.~Stoye.
\newblock Minimax regret treatment choice with finite samples.
\newblock \emph{Journal of Econometrics}, 151\penalty0 (1):\penalty0 70--81,
  2009.

\bibitem[Swaminathan and Joachims(2015)]{JMLR:v16:swaminathan15a}
A.~Swaminathan and T.~Joachims.
\newblock Batch learning from logged bandit feedback through counterfactual
  risk minimization.
\newblock \emph{Journal of Machine Learning Research}, 16:\penalty0 1731--1755,
  2015.

\bibitem[Tamar et~al.(2015)Tamar, Glassner, and
  Mannor]{Tamar:2015:OCV:2888116.2888133}
A.~Tamar, Y.~Glassner, and S.~Mannor.
\newblock Optimizing the {CV}a{R} via sampling.
\newblock In \emph{Proceedings of the Twenty-Ninth AAAI Conference on
  Artificial Intelligence}, AAAI'15, pages 2993--2999. AAAI Press, 2015.
\newblock ISBN 0-262-51129-0.

\bibitem[Tao and Wang(2017)]{tao2016adaptive}
Y.~Tao and L.~Wang.
\newblock Adaptive contrast weighted learning for multi-stage multi-treatment
  decision-making.
\newblock \emph{Biometrics}, 73\penalty0 (1):\penalty0 145--155, 2017.

\bibitem[Tetenov(2012)]{tetenov2012statistical}
A.~Tetenov.
\newblock Statistical treatment choice based on asymmetric minimax regret
  criteria.
\newblock \emph{Journal of Econometrics}, 166\penalty0 (1):\penalty0 157--165,
  2012.

\bibitem[Tian et~al.(2014)Tian, Alizadeh, Gentles, and
  Tibshirani]{tian2014simple}
L.~Tian, A.~A. Alizadeh, A.~J. Gentles, and R.~Tibshirani.
\newblock A simple method for estimating interactions between a treatment and a
  large number of covariates.
\newblock \emph{Journal of the American Statistical Association}, 109\penalty0
  (508):\penalty0 1517--1532, 2014.

\bibitem[Wang et~al.(2018)Wang, Zhou, Song, and Sherwood]{wang2017quantile}
L.~Wang, Y.~Zhou, R.~Song, and B.~Sherwood.
\newblock Quantile-optimal treatment regimes.
\newblock \emph{Journal of the American Statistical Association}, 113\penalty0
  (523):\penalty0 1243--1254, 2018.

\bibitem[Watkins(1989)]{watkins1989learning}
C.~J. C.~H. Watkins.
\newblock \emph{Learning from delayed rewards}.
\newblock PhD thesis, University of Cambridge England, 1989.

\bibitem[Xiao et~al.(2019)Xiao, Zhang, and Lu]{xiao2019robust}
W.~Xiao, H.~H. Zhang, and W.~Lu.
\newblock Robust regression for optimal individualized treatment rules.
\newblock \emph{Statistics in Medicine}, 38\penalty0 (11):\penalty0 2059--2073,
  2019.

\bibitem[Zhang et~al.(2012)Zhang, Tsiatis, Laber, and
  Davidian]{zhang2012robust}
B.~Zhang, A.~A. Tsiatis, E.~B. Laber, and M.~Davidian.
\newblock A robust method for estimating optimal treatment regimes.
\newblock \emph{Biometrics}, 68\penalty0 (4):\penalty0 1010--1018, 2012.

\bibitem[Zhang et~al.(2020)Zhang, Chen, Fu, He, Zhao, and Liu]{zhangss20}
C.~Zhang, J.~Chen, H.~Fu, X.~He, Y.~Zhao, and Y.~Liu.
\newblock Multicategory outcome weighted margin-based learning for estimating
  individualized treatment rules.
\newblock \emph{Statistica Sinica}, 20\penalty0 (4):\penalty0 1857--1879, 2020.

\bibitem[Zhang et~al.(2015)Zhang, Laber, Tsiatis, and Davidian]{zhang2015using}
Y.~Zhang, E.~B. Laber, A.~Tsiatis, and M.~Davidian.
\newblock Using decision lists to construct interpretable and parsimonious
  treatment regimes.
\newblock \emph{Biometrics}, 71\penalty0 (4):\penalty0 895--904, 2015.

\bibitem[Zhao et~al.(2012)Zhao, Zeng, Rush, and Kosorok]{zhao2012estimating}
Y.~Zhao, D.~Zeng, A.~J. Rush, and M.~R. Kosorok.
\newblock Estimating individualized treatment rules using outcome weighted
  learning.
\newblock \emph{Journal of the American Statistical Association}, 107\penalty0
  (499):\penalty0 1106--1118, 2012.

\bibitem[Zhao et~al.(2015)Zhao, Zeng, Laber, Song, Yuan, and
  Kosorok]{zhao2015doubly}
Y.-Q. Zhao, D.~Zeng, E.~B. Laber, R.~Song, M.~Yuan, and M.~R. Kosorok.
\newblock Doubly robust learning for estimating individualized treatment with
  censored data.
\newblock \emph{Biometrika}, 102\penalty0 (1):\penalty0 151, 2015.

\bibitem[Zhou et~al.(2017)Zhou, Mayer-Hamblett, Khan, and
  Kosorok]{zhou2017residual}
X.~Zhou, N.~Mayer-Hamblett, U.~Khan, and M.~R. Kosorok.
\newblock Residual weighted learning for estimating individualized treatment
  rules.
\newblock \emph{Journal of the American Statistical Association}, 112\penalty0
  (517):\penalty0 169--187, 2017.

\end{thebibliography}
\end{spacing}
\end{document}